\documentclass{article}

\usepackage{arxiv}

\usepackage[utf8]{inputenc} 
\usepackage[T1]{fontenc}    
\usepackage{hyperref}       
\usepackage{url}            
\usepackage{booktabs}       
\usepackage{amsfonts}       
\usepackage{nicefrac}       
\usepackage{microtype}      
\usepackage{lipsum}
\usepackage{graphicx}
\usepackage{multirow}
\usepackage{amsmath}
\usepackage{float}
\usepackage{subfig}
\usepackage{threeparttable}
\usepackage{array}
\usepackage{amsmath}
\usepackage[noend]{algpseudocode}
\usepackage{algorithmicx,algorithm}

\newtheorem{theorem}{Theorem}
\newtheorem{lemma}{Lemma}
\newtheorem{proof}{Proof}
\newtheorem{corollary}{Corollary}

\title{Epidemic Plateau: A Phenomenon under Adaptive Prevention Strategies}

\author{
 Hao~Liao \\
  Shenzhen University\\
  Shenzhen, China\\
  \texttt{jamesliao520@gmail.com} \\
  \And
  Ziqiang~Wu \\
  Shenzhen University\\
  Shenzhen, China\\
  \texttt{wzqgddg@foxmail.com} \\
   \And
 Alexandre~Vidmer \\
  Shenzhen University\\
  Shenzhen, China\\
  \texttt{alexandre@vidmer.com} \\
  \And
 Mingyang~Zhou \\
  Shenzhen University\\
  Shenzhen, China\\
  \texttt{zhoumy2010@gmail.com} \\
  \And
 Wei~Chen \\
  Microsoft Research Asia\\
  Beijing, China\\
  \texttt{weic@microsoft.com} \\
}

\begin{document}
\maketitle

\begin{abstract}
Since the beginning of the COVID-19 spreading, the number of studies on the epidemic models increased dramatically. It is important for policymakers to know how the disease will spread and what are the effects of the policies and environment on the spreading. In this paper, we propose two extensions to the standard SIR model: (a) we consider the prevention measures adopted based on the current severity of the infection. Those measures are adaptive and change over time; (b) multiple cities and regions are considered, with population movements between those cities and regions, while taking into account that each region may have different prevention measures. Although the adaptive measures and mobility of the population were often observed during the pandemic, these effects are rarely explicitly modeled and studied in the classical epidemic models. The model we propose gives rise to a plateau phenomenon: the number of people infected by the disease stays at the same level during an extended period of time. We show what are conditions need to be met in order for the spreading to exhibit a plateau period in a single city. In addition, this phenomenon is interdependent when considering multiple cities. We verify from the real-world data that the plateau phenomenon does exist in many regions of the world in the current COVID-19 development. Finally, we provide theoretical analysis on the plateau phenomenon for the single-city model and derive a series of results on the emergence and the ending of the plateau, as well as on the height and length of the plateau. Our theoretical results match well with our experimental findings.
\end{abstract}


\section{Introduction}

\label{Introduction}

Since December 2019, a virus named SARS-CoV-2 has been affecting the entire world drastically until this day. 
As of November 1, 2021, the virus had infected 246,594,191 people, including UK Prime Minister Boris Johnson and US President Donald Trump. 
Although the epidemic in some countries has been brought under control, the epidemic in most parts of the world is still ongoing. 
Nowadays (November 2021), the vaccination rate of COVID-19 vaccines is increasing, but there are still hundreds of thousands of new cases per day.

The number of studies based on the epidemic analysis has considerably grown since the start of the COVID-19 epidemic.
Some are about the characteristics of COVID-19. Although most of the attributes of SARS-CoV-2 cannot be obtained by epidemic analysis, the epidemiological characteristics of COVID-19 can be estimated. Most of the works~\cite{wu2020nowcasting, hao2020reconstruction, cao2020incorporating, deressa2020modeling} focus on the study of the basic reproduction number $R_0$ of COVID-19, which can reflect the propagation ability of SARS-CoV-2, while other works analyze the characteristics of COVID-19 such as the generation interval (cases doubling time)~\cite{guan2020clinical, ali2020serial}, the duration~\cite{walsh2020duration} and the infectiousness~\cite{article}. 
Others are about the trend forecast of COVID-19. By combining real data with the model, these works aim to predict the development of COVID-19. Since the beginning of the COVID-19 outbreak, it has been a hot topic. The most common method to predict the number of future infections is to fit real-time data, such as propagation models~\cite{ghamizi2020data, ren2020derivation} and Bayes~\cite{gaglione2020adaptive, de2020epidemiological}. Due to the global characteristics of the epidemic, many focus on combining population mobility data with the propagation model~\cite{2020Understanding, jia2020population, huang2020understanding}.

SIR~\cite{newman10} is a widely used model, and its related models have been well studied in the literature.
Apart from the combination of the population mobility, there are additional aspects that are studied: The SEAIRD model allows for further division of the population~\cite{liu2020towards, nino2019two}; environmental parameters are added to the models~\cite{zhang2017forecasting, Bherwani2020}; and machine learning algorithms are used~\cite{zhang2020predicting, achterberg2020comparing}.

It is worth noting that the significance of prevention has also become the subject of many studies, including reducing population mobility~\cite{jia2020population}, lockdown of the city~\cite{tian2020investigation}, keeping social distance~\cite{du2020effects}, and vaccination~\cite{2021Vaccine,2021Time}.
Based on the development of the epidemic situation before and after prevention, the role of prevention is analyzed. 
Almost all the results affirm the role of prevention and point out that the epidemic would worsen more seriously without it.

The above works have studied various aspects of on COVID-19 pandemic, but rare work takes the effect of adaptive changes of prevention measures over time on the epidemic development into account.
During the COVID-19 pandemic, many regions gradually adopt adaptive prevention strategies to try to balance epidemic prevention and economic growth. Different regions often employ different prevention strategies, which may bring new phenomenon that does not exist in the traditional epidemic models.
We propose two extensions to the widely used SIR model, considering adaptive changes in prevention strategies based on the current infection severity and the population mobility between multiple regions.
We show the plateau phenomenon for the first time to address both the theoretical and the empirical perspectives.
We verify from the real-world data that the plateau phenomenon does exist in many regions of the world in the current COVID-19 development, indicating that this is an important phenomenon to study.

This work provides new insights into the development trend of COVID-19 and its prevention in cities. 
The development of the epidemic not only has ups and downs but also can have a plateau period of steady development. 
Under the adaptive prevention strategies of the city, we may be able to control the epidemic at a low cost and achieve long-term coexistence with the virus. 
This is a new option for outbreaks that would not fade away in a short period of time.

\section{Results}
To make it easier to understand our results, we summarize the symbols used in this paper, as shown in Table~\ref{tab:symbol}. 
We show our findings below from three aspects: simulation, the real world, and theory.

\begin{table*}[t]
    \setlength{\abovecaptionskip}{2mm}
	\centering
    \caption{A brief explanation of symbols.}
	\small
	\label{tab:symbol}
	\setlength{\tabcolsep}{2mm}
	\begin{threeparttable}				
		\begin{tabular}{l|c|c|c}
			\toprule
			\textbf{symbol} &  \textbf{explanation}  &  \textbf{symbol} &  \textbf{explanation} \\
			\midrule 
			 S &   the number of susceptible people  &   I  &   the number of infected people \\
			\midrule 
			 R &   the number of recovered people  &   N  &   total number of people \\
			\midrule 
			 $\beta$ & infection rate  &  $\gamma$  &   recovery rate \\
			 \midrule 
			 $r$ & the number of contacts; prevention parameter  &  t  &   day \\
			 \midrule 
			 $\theta$ & prevention threshold  &  k  &   reaction speed \\
			 \midrule 
			 $[r_L, r_H]$ & the range of changing prevention parameter  &  $[p_L, p_H]$  &   the range of changing proportion of floating population \\
			 \midrule 
			 $r_0$ & initial prevention parameter  &  $p_0$  &   initial proportion of floating population \\
			 \midrule 
			 $\tilde{r}_t$ & effective prevention parameter  &  $I_0$  &   the initial number of infected people \\
			 \midrule 
			 IC & initial condition of the city $(N, \frac{I_0}{N})$  &  EE  &   external environment \\
			\bottomrule
		\end{tabular}
	\end{threeparttable}
\end{table*}

\subsection{Simulation Results}
\label{simulationexperiment}

Based on the epidemic model with the prevention parameter we proposed, we take the COVID-19 as an example to simulate the epidemic development in a single city and in multiple cities, which use one day as the time unit.

\subsubsection{Prevention Strategy}
\label{sec:ps}

During the epidemic, different cities or regions may have different strategies and measures to deal with infectious diseases.
For example, during the COVID-19 outbreak, some regions imposed mandatory mask wearings while some did not, many enforce social distancing,
	but at different scales, such as limited the size of the gathering of no more than 50 or 100 people, closing schools, or restricting
	travels to and from high-risk regions.
All these restrictions can be reflected in two parameter changes within our model.

First is the prevention parameter $r_i$ for each city $i$.
Different social distancing measures essentially all reduce the number of people any individual would meet in a time unit, which is
	modeled by $r_i$.
Thus we use $r_i$ to model the prevention measures, with lower $r_i$ meaning a high level of prevention measures and high $r_i$ meaning
	a low level of prevention measures.
Another parameter affected by the prevention measure is the fraction of 
	the floating population $p_i$ --- the proportion of the floating population in the urban population.
A higher prevention level decreases the fraction of the floating population, while a lower prevention level increases the fraction of the floating population.

Most cities or regions also dynamically adapt their prevention measures based on the current infection data --- when the new infection cases
	rise, the prevention measures intensify, and when the number of new cases drops, the prevention measures are relaxed.
We model this scenario by setting a prevention threshold and dynamically changing the prevention parameter $r_i$ and 
	the fraction of the floating population $p_i$
	based on whether the current infection cases are below or above the prevention threshold.
Technically, we have three main factors for deciding and changing the prevention measures:

\begin{itemize}
\item 
{\bf Prevention threshold}, which sets a limit on the number of new infections within a time period.
New infections only include new cases caused by intra-city infections, excluding imported cases brought outside the city.
When the new number of infected people each day for 14 consecutive days is lower than the prevention threshold, it shows that the epidemic situation is getting under control, and the prevention level begins to decline ($r$ increases). 
When the number of newly infected people on any given day is higher than the prevention threshold, it shows that the epidemic may worsen again, and the prevention level begins to rise ($r$ decreases).

\item 
{\bf The range of changing prevention parameters}, which gives the range of changes
	of each prevention parameter $r_i$ and the fraction of floating population parameter $p_i$.
	When the prevention measure intensifies, the values are changed to $r_{L}$ and $p_{L}$ respectively, and when the prevention
	measure relaxes, the values are changed to $r_{H}$ and $p_{H}$.

\item 
{\bf The reaction speed}, which measured by the number of days $k$ to complete the change.
In the real world, the change of prevention measures never happens instantaneously. 
So we use $k$ to model the number of days needed to change $r_{L}$ to $r_{H}$ and $p_{L}$ to $p_{H}$, or in the reverse order.
We use $r_{t+1}=r_{t} \cdot (\frac{r_{H}}{r_{L}})^{\frac{(-1)^{s}}{k}}$ and $p_{t+1}=p_{t} \cdot (\frac{p_{H}}{p_{L}})^{\frac{(-1)^{s}}{k}}$ to represent daily changes in $r$ and $p$.
When $s=0$, it means to relax the prevention. 
When $s=1$, it means to strengthen prevention.
$r$ and $p$ are always within the range of the prevention strategy.

\end{itemize}

In general, prevention strategies can be divided into two cases.
The case that the prevention parameter remains unchanged: the parameters of prevention parameter $r$ and the proportion of floating population $p$ are fixed values, and there is no prevention threshold and reaction speed.
This corresponds to the classical infectious disease model, and we refer to it as a model with a {\em constant prevention strategy}.
The case corresponds to the adaptive changes in prevention parameter that we commonly see during the COVID-19 pandemic:
	the full strategy is characterized by the prevention threshold $\theta$, the $r$-range $[r_{L},r_{H}]$, the
	$p$-range $[p_{L},p_{H}]$, and the reaction speed $k$, and we refer to such ones as models with an {\em adaptive prevention strategy}.

\begin{figure*}[t]

	\small
	\centering     
	\subfloat[\scriptsize  Single city with constant strategy ]{\label{singlecitywithoutstrategy}\includegraphics[width=0.48\textwidth]{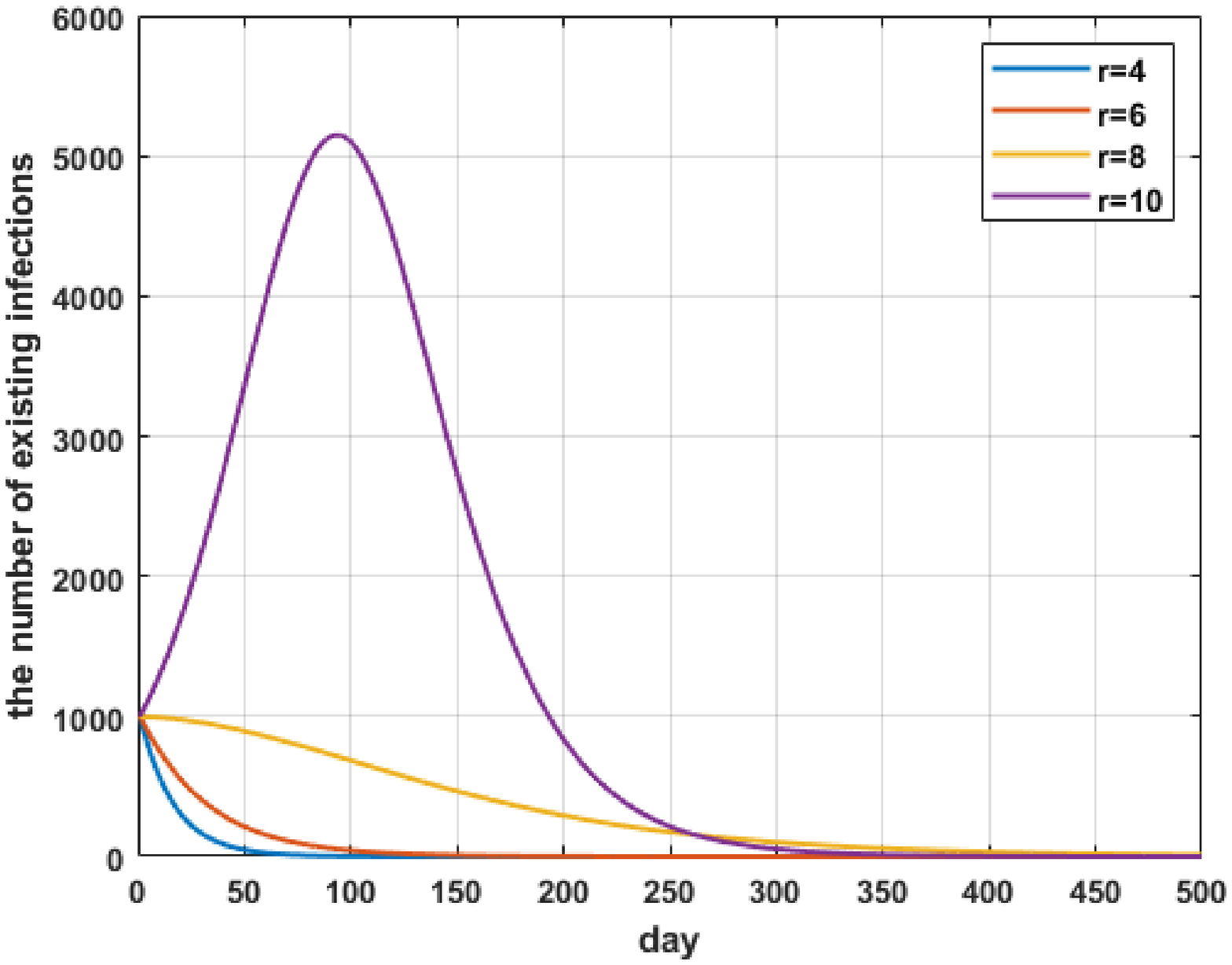}}	
	\subfloat[\scriptsize  Single city with adaptive strategy ]{\label{singlecitywithstrategy}\includegraphics[width=0.48\textwidth]{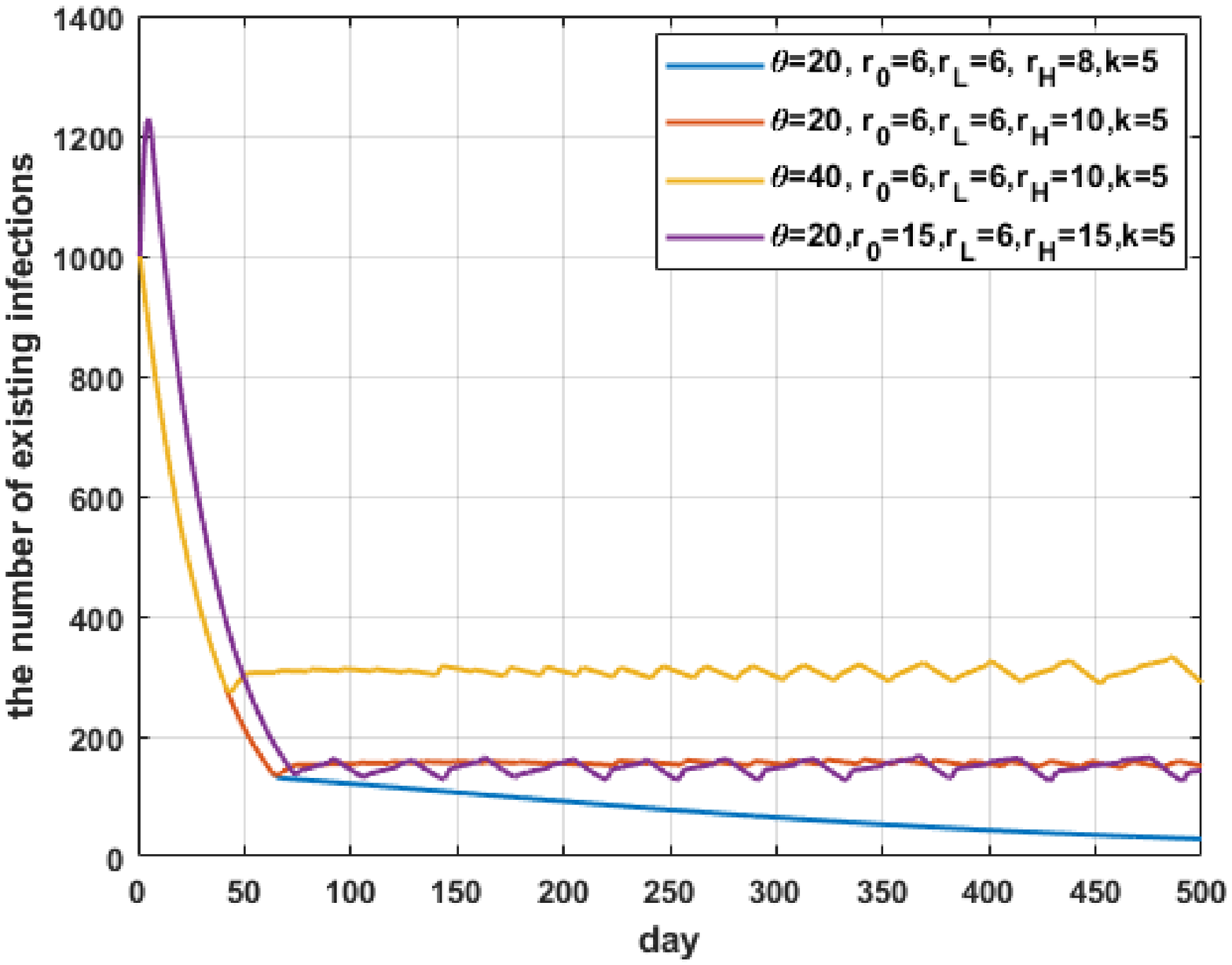}}

	\subfloat[\scriptsize Multiple cities with constant strategies ]{\label{multiplecitieswithoutstrategy}\includegraphics[width=0.48\textwidth]{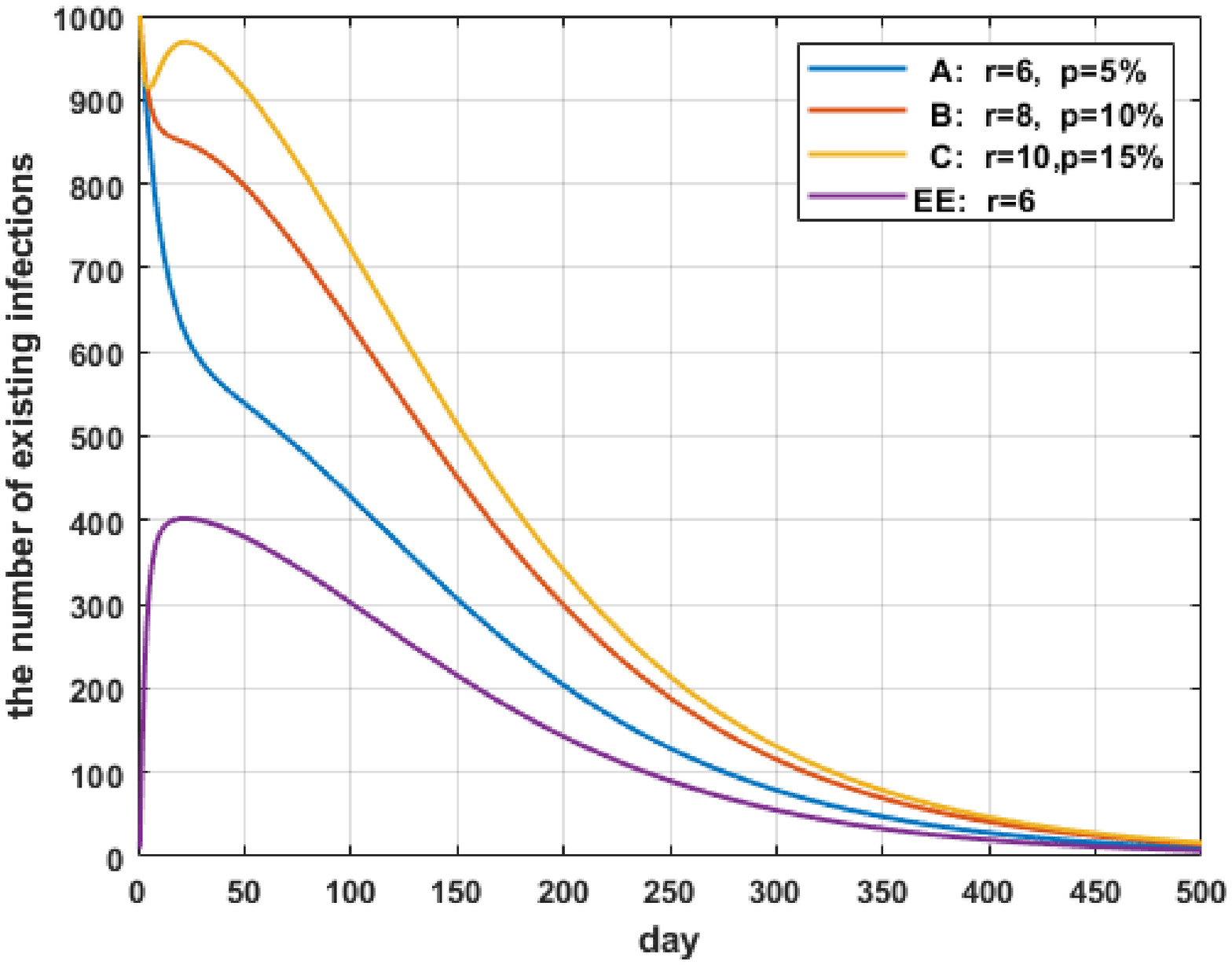}}
	\subfloat[\scriptsize  Multiple cities with mixed strategies ]{\label{multiplecitieswithstrategy}\includegraphics[width=0.48\textwidth]{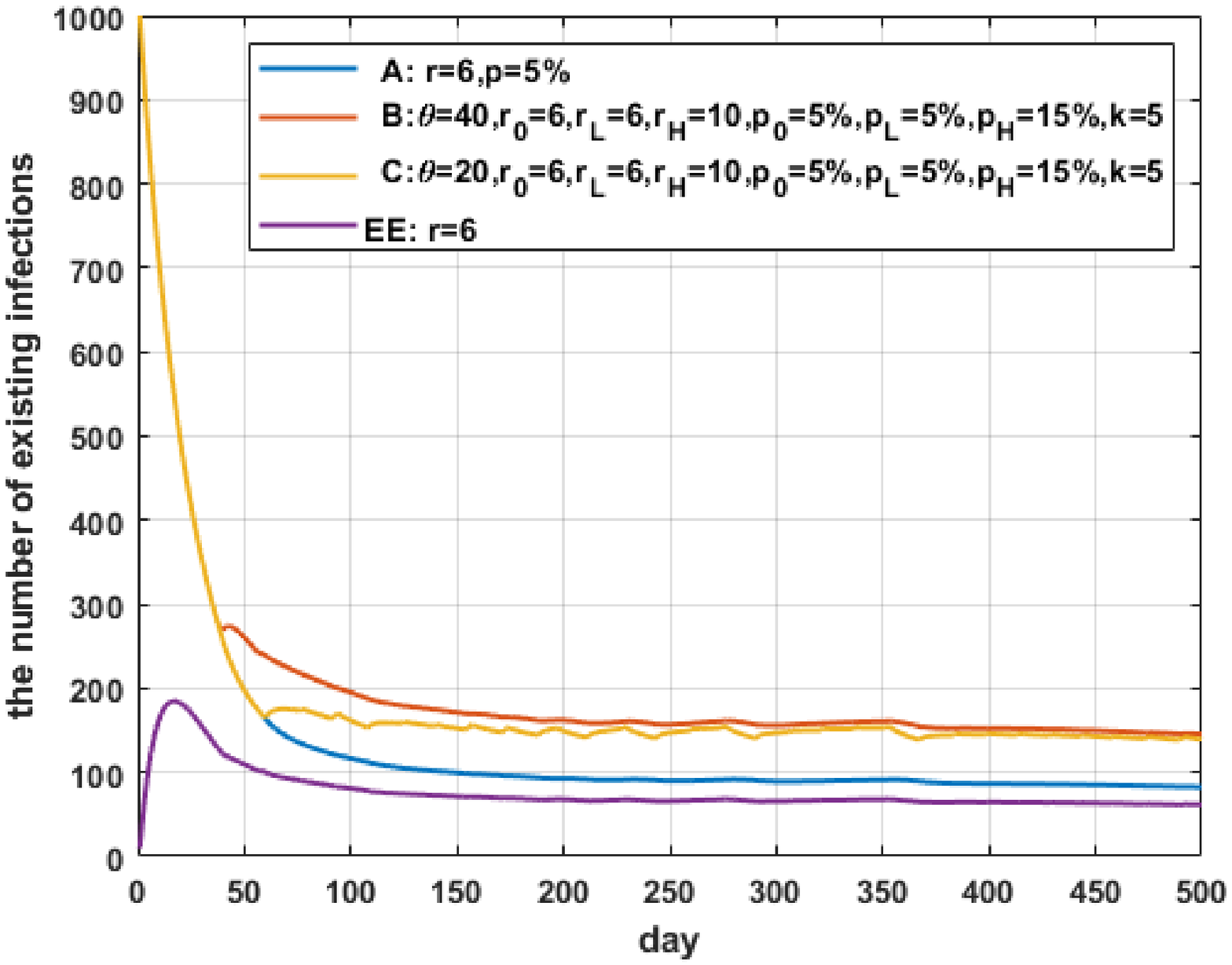}}
	\caption{Infectious disease model with constant or adaptive prevention strategies.}
	\label{plateau}

\end{figure*}

Figure~\ref{plateau} shows the development of the epidemic in a single city and multiple cities with constant or adaptive 
	prevention strategies.
The y-axis is the value of $I$ for a particular day, which is the number of current existing infected individuals that can still transmit
	the disease (note that it is not the number of newly infected cases on that day).
As a convention, for a constant prevention strategy that does not change $r$, e.g. $r$ is always $6$, we denote it as $r=6$ in the 
figure legend; for an adaptive prevention strategy, 
	if we initially set a prevention strategy to be at a certain level, and then it will change 
	in the range $[r_L, r_H]$ when the newly infected case per day drops below the threshold, we denote it as $r_0=6$ and 
	$r_L,r_H$ to the appropriate values in the legend.
The fraction of floating population $p$ and the corresponding $p_0$ and $p_L, p_H$ are treated in the same manner.
The initial conditions of cities are all $IC = (200000,0.5\%)$, while the external environment is $IC = (100000,0.01\%)$.
\begin{itemize}
\item
Figure~\ref{singlecitywithoutstrategy}: In the single-city model with a constant prevention strategy, the change curve of the number of existing infections is smooth.
The number of infections either continues to decline ($r \le 8$), or rises first and then falls ($r>8$).
\item
Figure~\ref{singlecitywithstrategy}: In the single-city model with an adaptive prevention strategy, the change curve of the number of existing infections is no longer smooth. 
Compared to the single-city model with a constant prevention strategy (Figure~\ref{singlecitywithoutstrategy}), it has a distinctive feature ---
	a long period in which the number of existing infected people remains more or less the same, with small fluctuations.
We refer to such a flat period as a \textbf{plateau}, and it is the main phenomenon we are going to study in this paper. 
\item
Figure~\ref{multiplecitieswithoutstrategy}: This is the case in three cities that have constant prevention strategies. 
In the multi-city model, it can be found that the urban epidemic situation affects each other under the influence of population mobility.
The number of infections in city C should have risen to a higher level ($I>5000$), but instead, it begins to decline earlier under the influence of other cities. 
On the contrary, there has been an increase in the number of infections in the external environment EE which should have been declining ($r=6$).
There is no plateau phenomenon in this case.
\item
Figure~\ref{multiplecitieswithstrategy}: In this case, two of the cities (B and C) adopt adaptive prevention strategies, and one city (A) as well
	as the external environment use a constant prevention strategy.
The phenomenon of horizontal fluctuation in the number of existing infections, i.e. the plateau, appears again.
More interestingly, the plateau appears in every city,  has occurred in every city and the environment, even if some cities and the environment
	adopt a constant prevention strategy.
\end{itemize}

\subsubsection{Plateau}
We can see that the adoption of adaptive prevention strategies is strongly correlated to the 
	emergence of the plateau.
When there is no adaptive prevention strategy deployed, we do not see the plateau phenomenon(Figure~\ref{singlecitywithoutstrategy} and Figure~\ref{multiplecitieswithoutstrategy}). 
However, the adoption of adaptive strategies itself is not sufficient to guarantee the appearance of the plateau, which can be
	seen in the blue curve in Figure~\ref{singlecitywithstrategy}.
Moreover, in the multi-city situation, even if some city does not adopt an adaptive prevention strategy, it still exhibit a plateau
	(city A in Figure~\ref{multiplecitieswithstrategy}), seemingly under the influence of other cities with adaptive prevention strategies.
Thus, the actual condition of the emergence of plateau needs to be further analyzed.

\begin{figure*}[t]

	\small
	\centering     
	\subfloat[\scriptsize ]{\label{threshold}\includegraphics[width=0.33\textwidth]{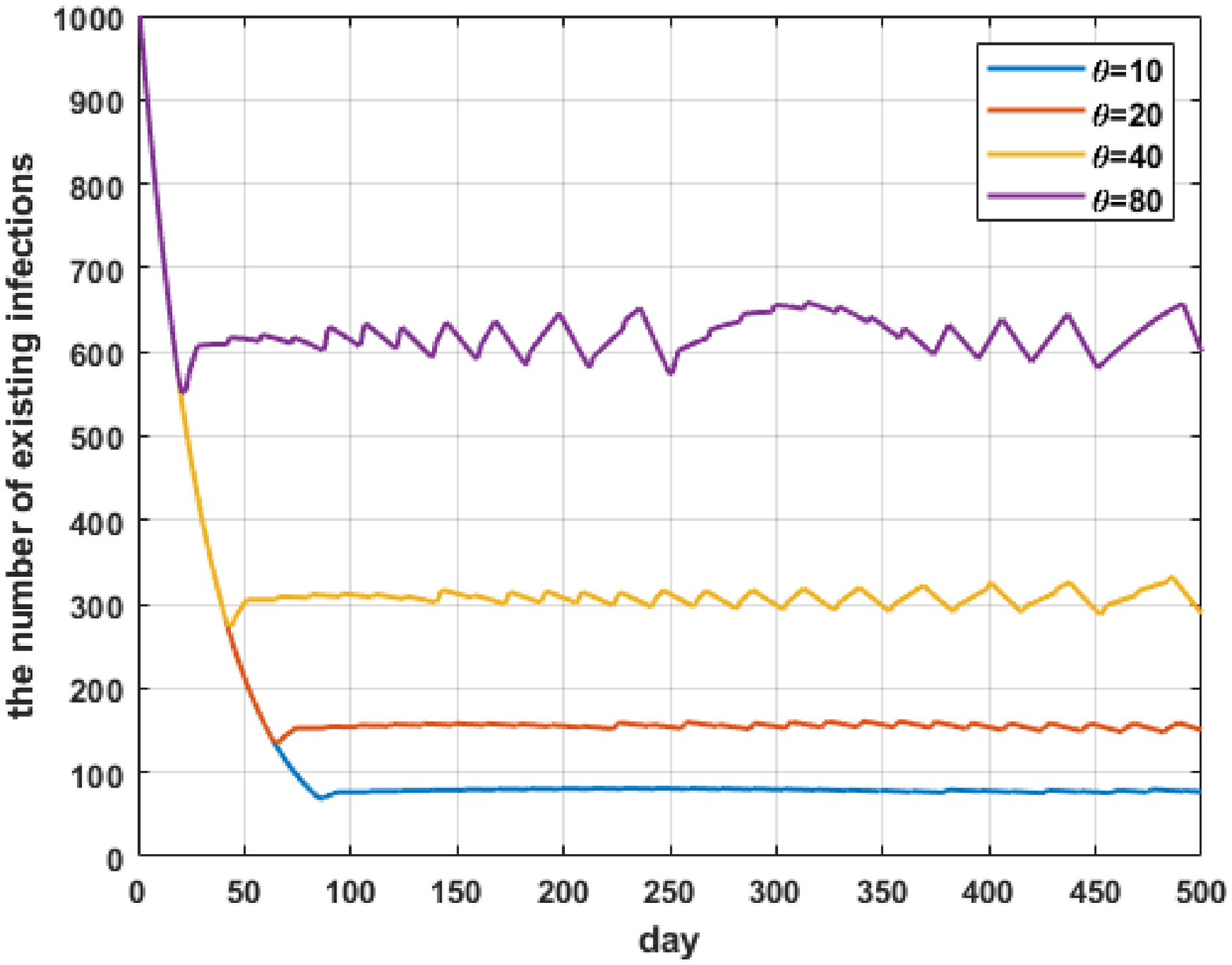}}	
	\subfloat[\scriptsize ]{\label{changingrange}\includegraphics[width=0.33\textwidth]{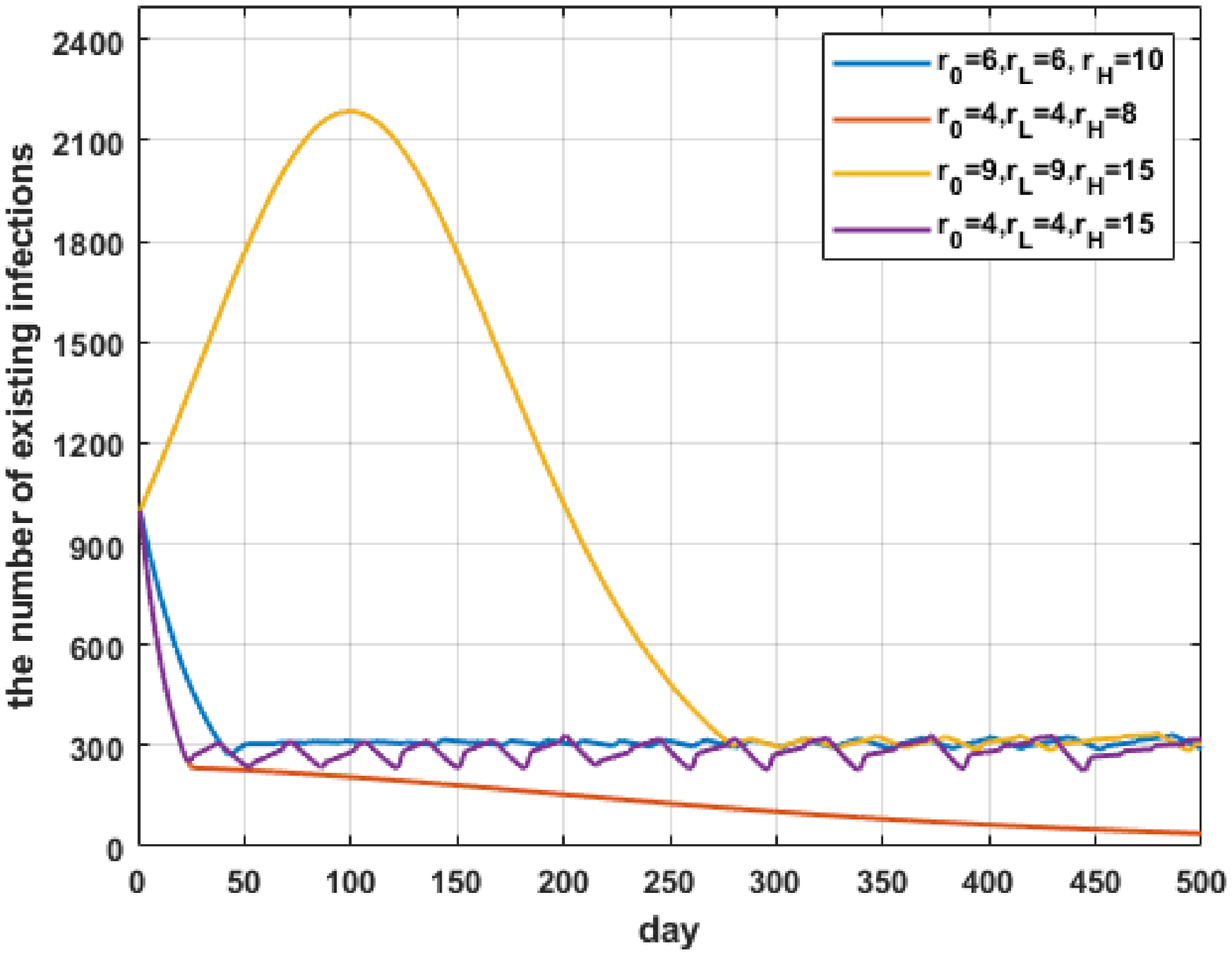}}
	\subfloat[\scriptsize ]{\label{changerate}\includegraphics[width=0.33\textwidth]{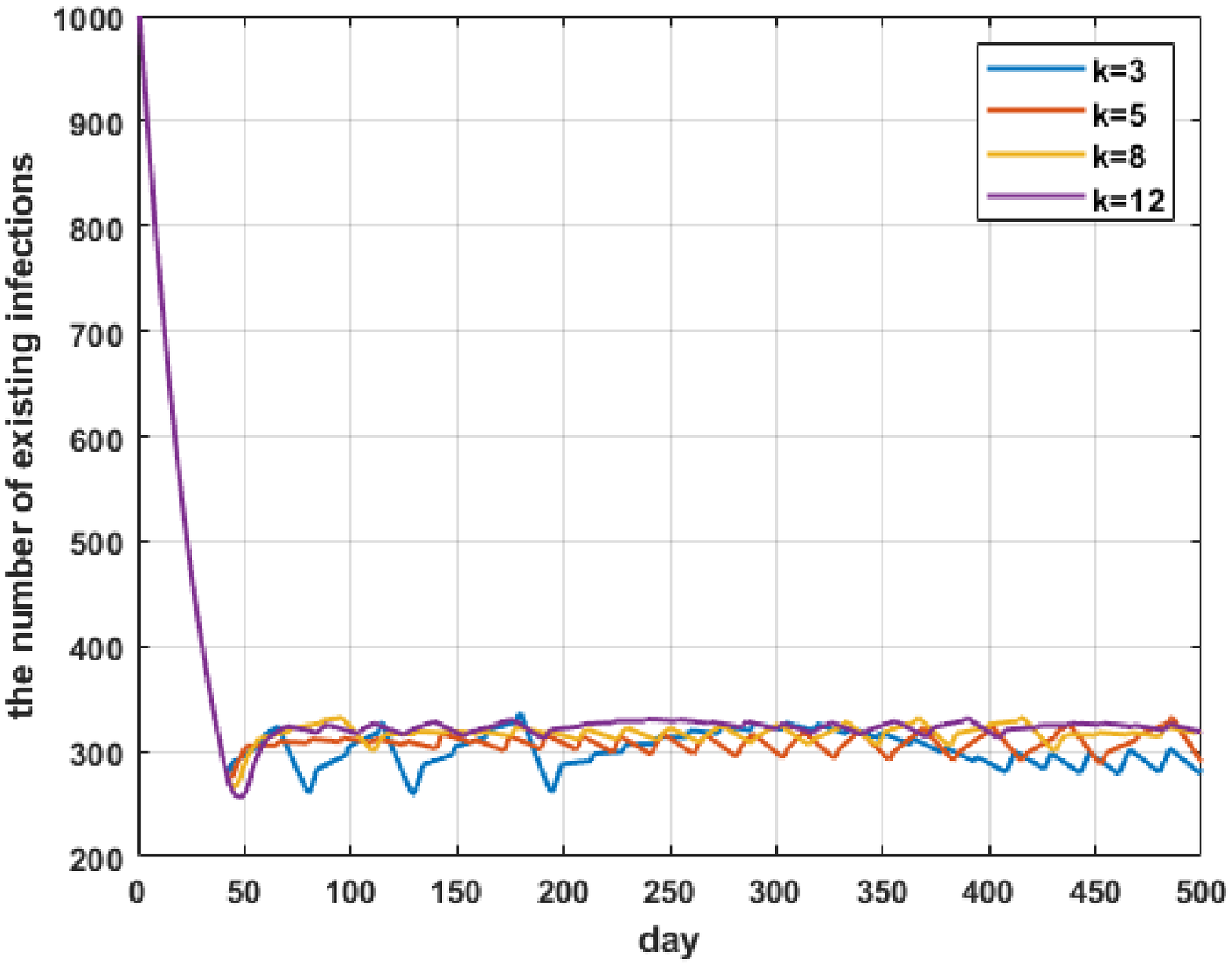}}
	
	\caption{The influence of indicators on the plateau.}
	\label{strategyinfluence}

\end{figure*}

\subsubsection{The Effect of Prevention Strategies}

We use the single-city model for this test to make the effect of prevention strategies more direct, as shown in Figure~\ref{strategyinfluence}: 
\begin{itemize}
    \item \textbf{Prevention threshold $\theta$.}
    As shown in Figure~\ref{threshold}, when there is a plateau in the development of the epidemic, it can be found that the prevention threshold $\theta$ is positively correlated with the height of the plateau. The greater the prevention threshold, the higher the plateau.

    \item \textbf{The range of changing the prevention parameter $[r_L,r_H]$.}
    As shown in Figure~\ref{changingrange}, when there is a plateau in the development of the epidemic, different $[r_L,r_H]$ will not change the height of the plateau.
    However, $[r_L,r_H]$ will determine whether the plateau occurs.
    When $r_0=4, r_L=4, r_H=8$, the plateau does not occur, and the number of existing infections continues is declining.
    When $r_0=9, r_L=9, r_H=15$, the number of existing infections will rise at first and then decrease, but there will still be a plateau after.
    A closer look at the blue line and the purple line, the size of the changing range is related to the amplitude.
    The greater the range of change, the greater the amplitude.

    \item \textbf{The reaction speed $k$.}
    Figure~\ref{changerate} shows that $k$ affects the fluctuation amplitude of the plateau. The smaller the $k$, the greater the fluctuation of the number of existing infections during the plateau.
\end{itemize}

Through our simulations, we can see that the emergence of the plateau is related to the range of changing the prevention parameter $[r_L,r_H]$, and the height of the plateau is positively correlated with prevention threshold $\theta$. It shows that the integration of prevention strategies, especially the adaptive ones, would change the dynamics of the epidemic that is not covered by the traditional epidemic models.

\begin{figure*}[t]

	\small
	\centering     
	\subfloat[\scriptsize  ]{\label{multiplecitiesplateau}\includegraphics[width=0.33\textwidth]{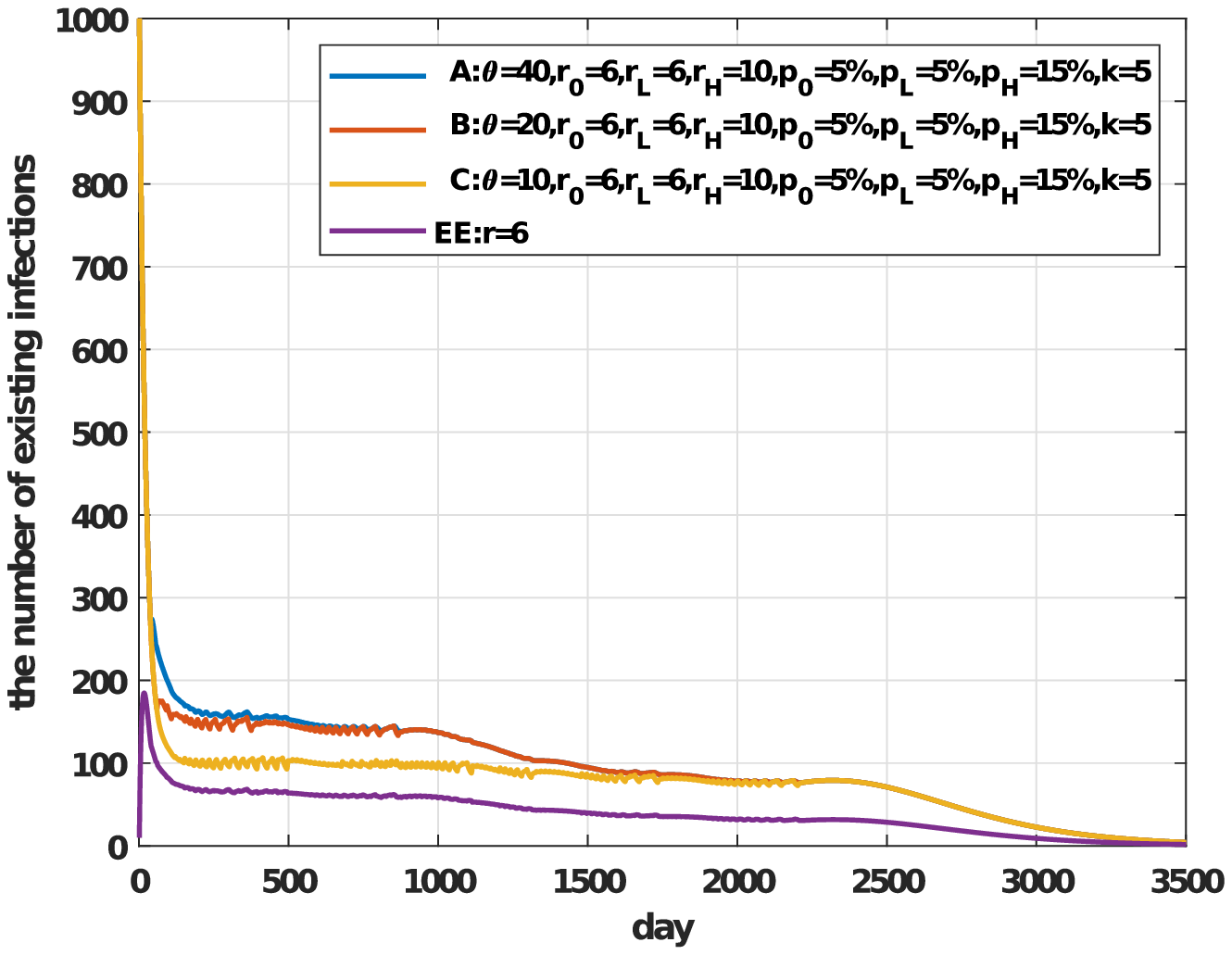}}	
	\subfloat[\scriptsize
	]{\label{onecitywithplateau}\includegraphics[width=0.33\textwidth]{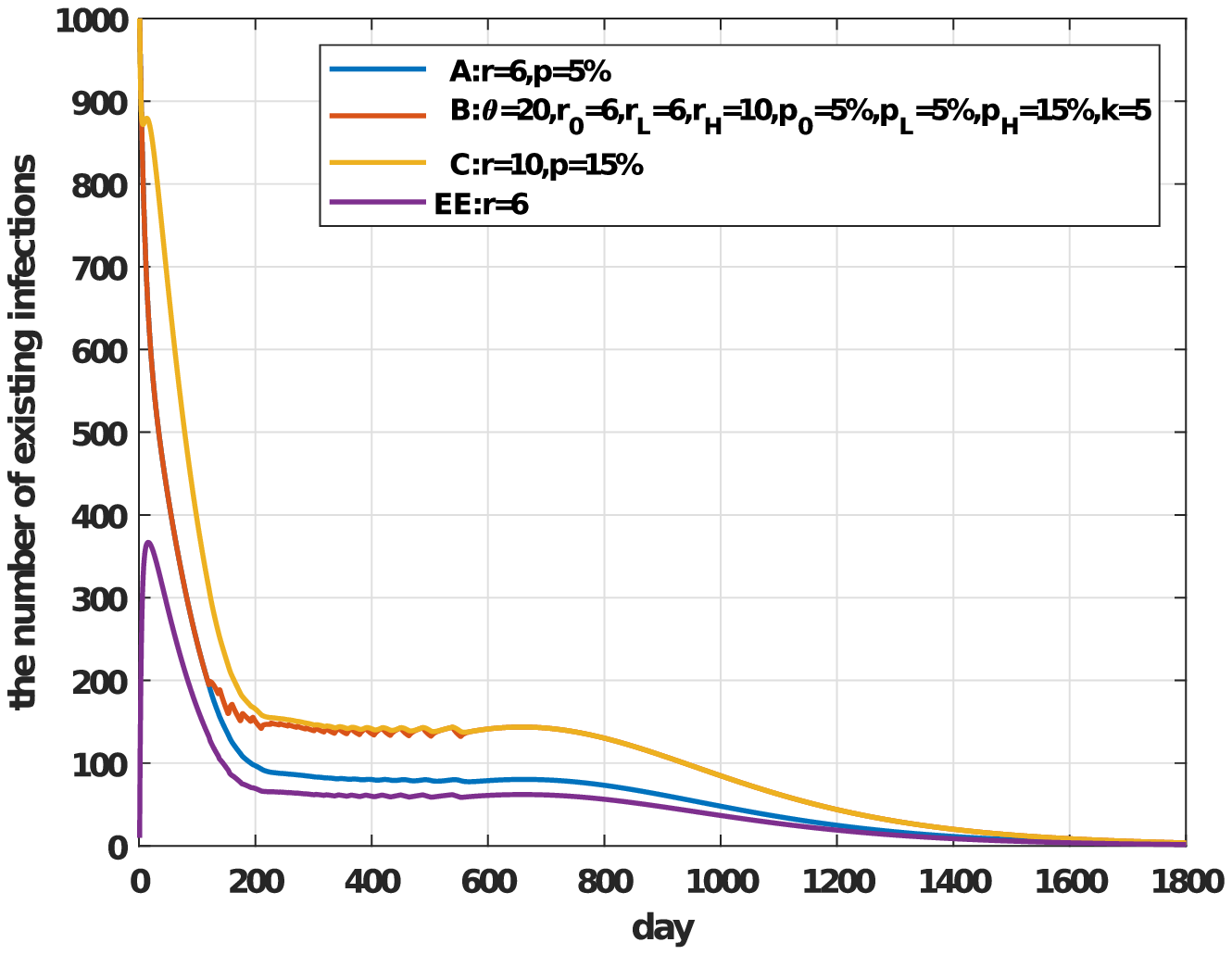}}
	\subfloat[\scriptsize ]{\label{noplateau}\includegraphics[width=0.33\textwidth]{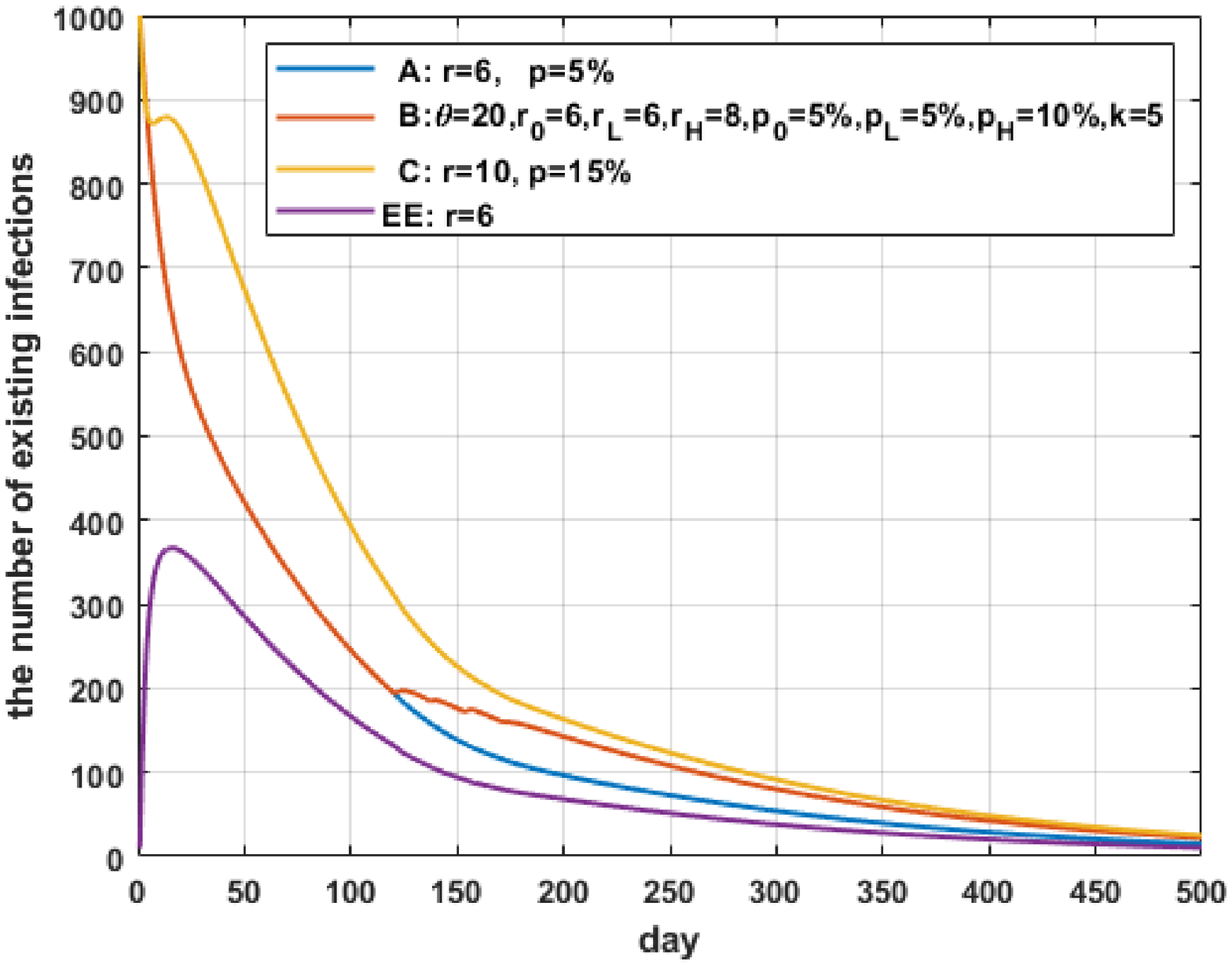}}
	
	\caption{The plateau phenomenon under the multi-city model.}
	\label{multistrategyinfluence}

\end{figure*}

As shown in Figure~\ref{multistrategyinfluence}, it is found that the plateau also exists in the multi-city model:
\begin{itemize}
    \item \textbf{Changes in the plateau of cities.} In the multi-city model, if all cities have the same $IC$ and adopt adaptive prevention strategies, they will all appear to plateau.
    The plateau of cities A, B and C end differently. City A is the earliest, and city C is the latest.
    However, when $I$ of city A drops to the same as city B, it will synchronize with city B and enter the plateau, as shown in Figure~\ref{multiplecitiesplateau}.
    When the plateau of city B (and A) ends, it will be synchronized with city C.
    The emergence of the synchronization phenomenon is due to the same prevention intensity among cities and the balance of the flow of infected people.
    \item \textbf{Cities with constant prevention strategies could also appear to plateau in the multi-city model.} In Figure~\ref{onecitywithplateau}, both city A and city C adopt a constant prevention strategy, in which the number of infections in city A should rise at first and then decline, and the number of infections in city C should be declining all the time.    However, under the influence of population mobility by city B, all cities have appeared plateau. Thus it can be seen that the emergence of the plateau in the multi-city model is also related to population mobility.
    \item \textbf{The cities that would have appeared to plateau did not appear that in the multi-city model.} As shown in Figure~\ref{noplateau}, city B is supposed to have a plateau, while city A and C, one will go down, the other will go up. However, all three cities are in decline. City A has been rising for some time, but obviously falling back ahead of time. This further highlights the impact of population mobility in the multi-city model.
\end{itemize}

The above experimental results show that quite an interesting plateau phenomenon emerges even in the multi-city model with complex epidemic dynamics within cities and population mobility dynamics across cities.
In particular, the results strongly indicate that all cities are interdependent in terms of epidemic development --- their plateau phenomenon co-occur, and even a city that would not have a plateau by itself would appear to plateau under the influence of other cities with adaptive strategies.
This suggests that in battling the COVID-19 pandemic, one should not look only at the local prevention strategies, but also need to take into account other regions' strategies and population mobility, and perhaps some global coordination would be more effective.

\subsection{Empirical Results}
\label{sec:realworld}

In fact, the plateau phenomenon also appears in the real world epidemic, in particular in the current COVID-19 pandemic. 
We collect the epidemic data as of October 25, 2021, from data source DXY\footnote{https://ncov.dxy.cn/}, which collected data from the World Health Organization (WHO), and China National Health Commission.
We have carried on the plateau fitting to the data of 213 countries (through the coefficient of determination) and found that 77 countries have obvious plateau phenomenon in more than one month~\footnote{https://github.com/Complex-data/COVID-19\_Data}. This is double the 35 countries we counted before (October 2020).
The coefficient of determination $D^2$ is calculated as follows:

\begin{align*}
D^2 & = 1- \frac{\sum_{i=1}^{n}(y_i-Y_i)^2}{\sum_{i=1}^{n}(Y_i-\overline{Y})^2},
\end{align*}
where $y_i$ is the number of existing infections of simulation, $Y_i$ is the number of existing infections of real data, and $\overline{Y}$ is the average number of existing infections of real data. $D^2 \in (-\infty, 1]$, the closer the value of $D^2$ is to $1$, the better the degree of the fitting.

Among other countries, 113 of them are still suffering worsening situations or are facing a new rising of the epidemics, while 22 countries are or have subsided.

\begin{figure*}[!tpb]
	\centering

	\subfloat[\scriptsize Poland ($D^2=0.9074$)]{\includegraphics[width=0.48\textwidth]{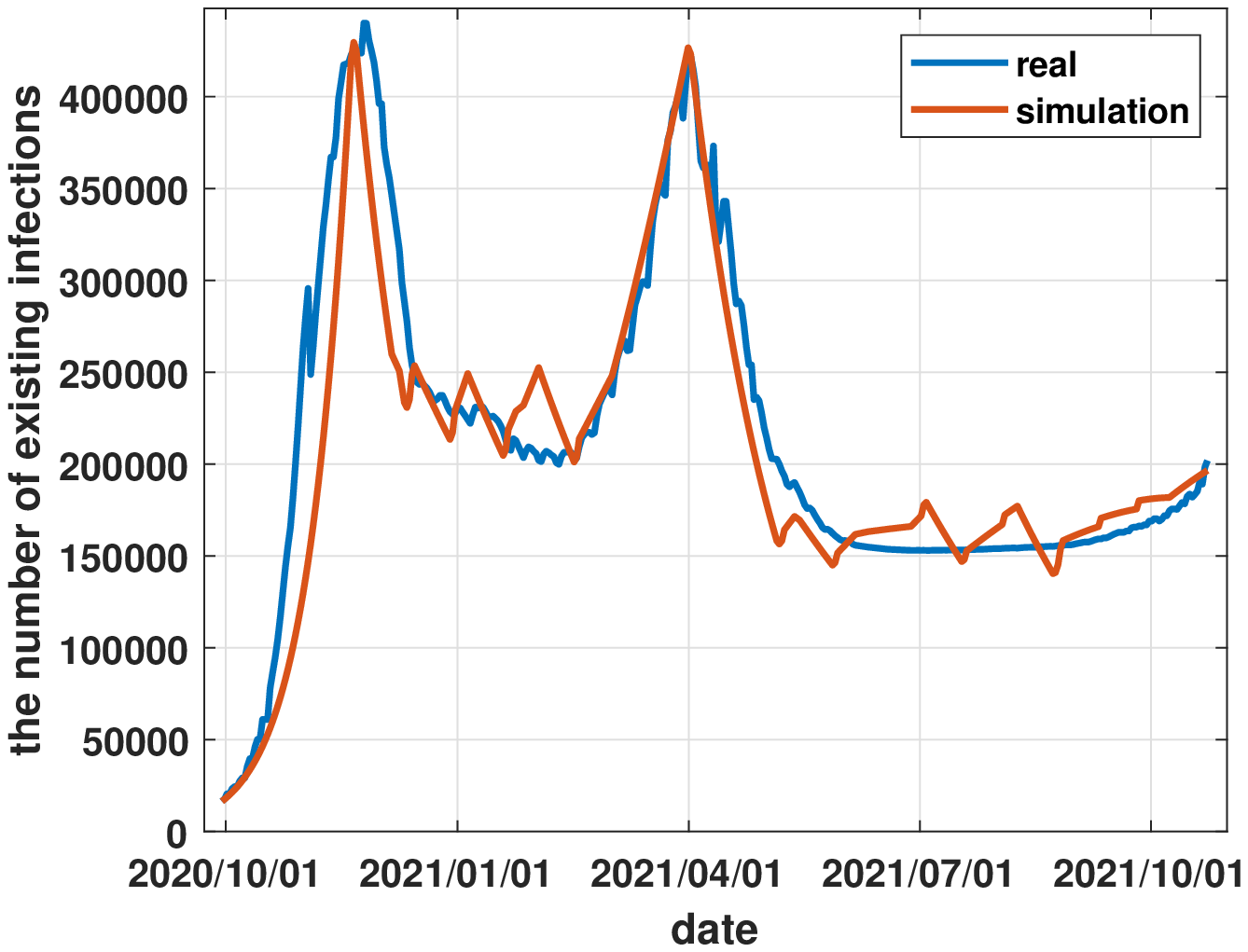}}
	\subfloat[\scriptsize Philippines ($D^2=0.9315$)]{\includegraphics[width=0.48\textwidth]{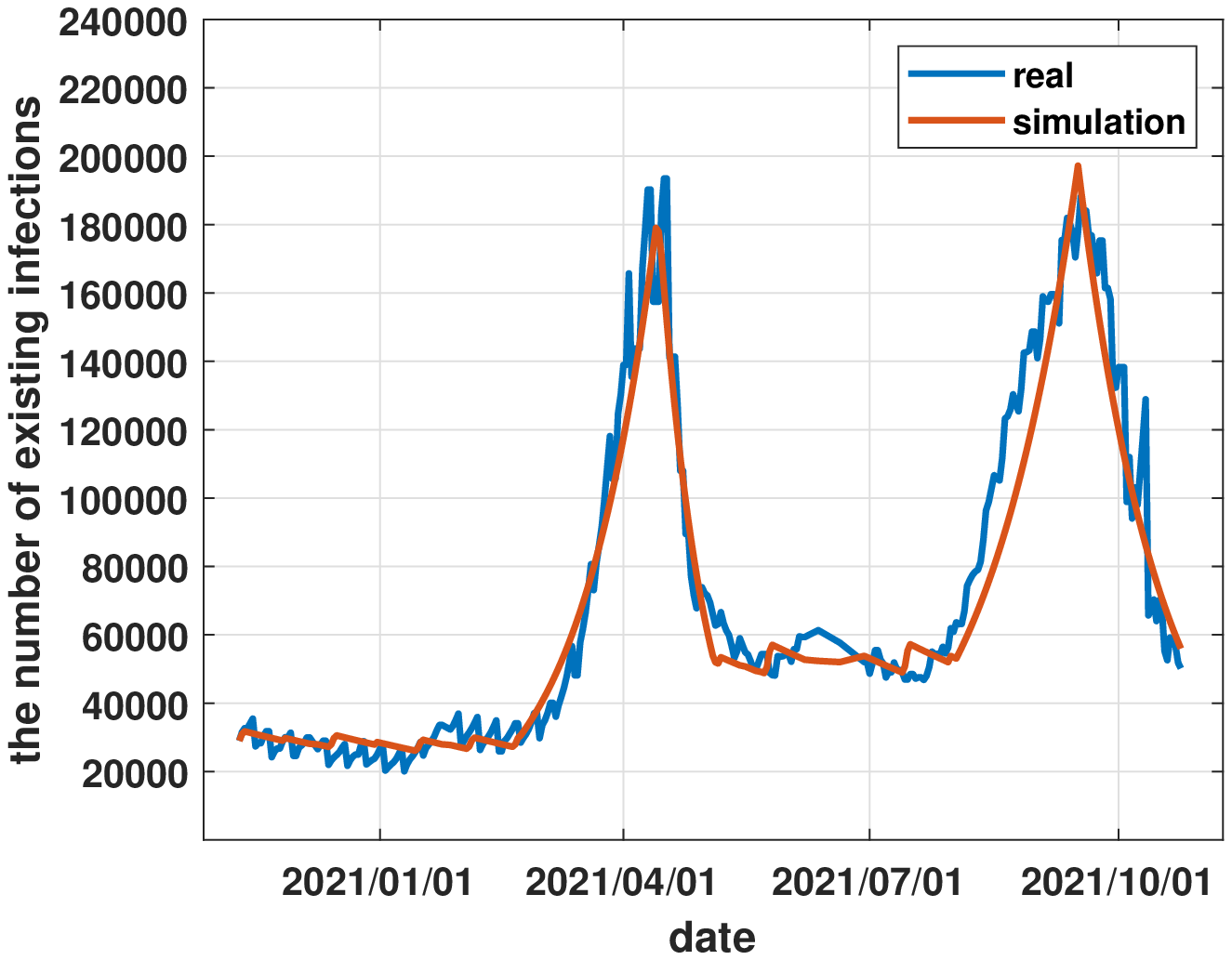}}

	\subfloat[\scriptsize Portugal ($D^2=0.9244$)]{\includegraphics[width=0.48\textwidth]{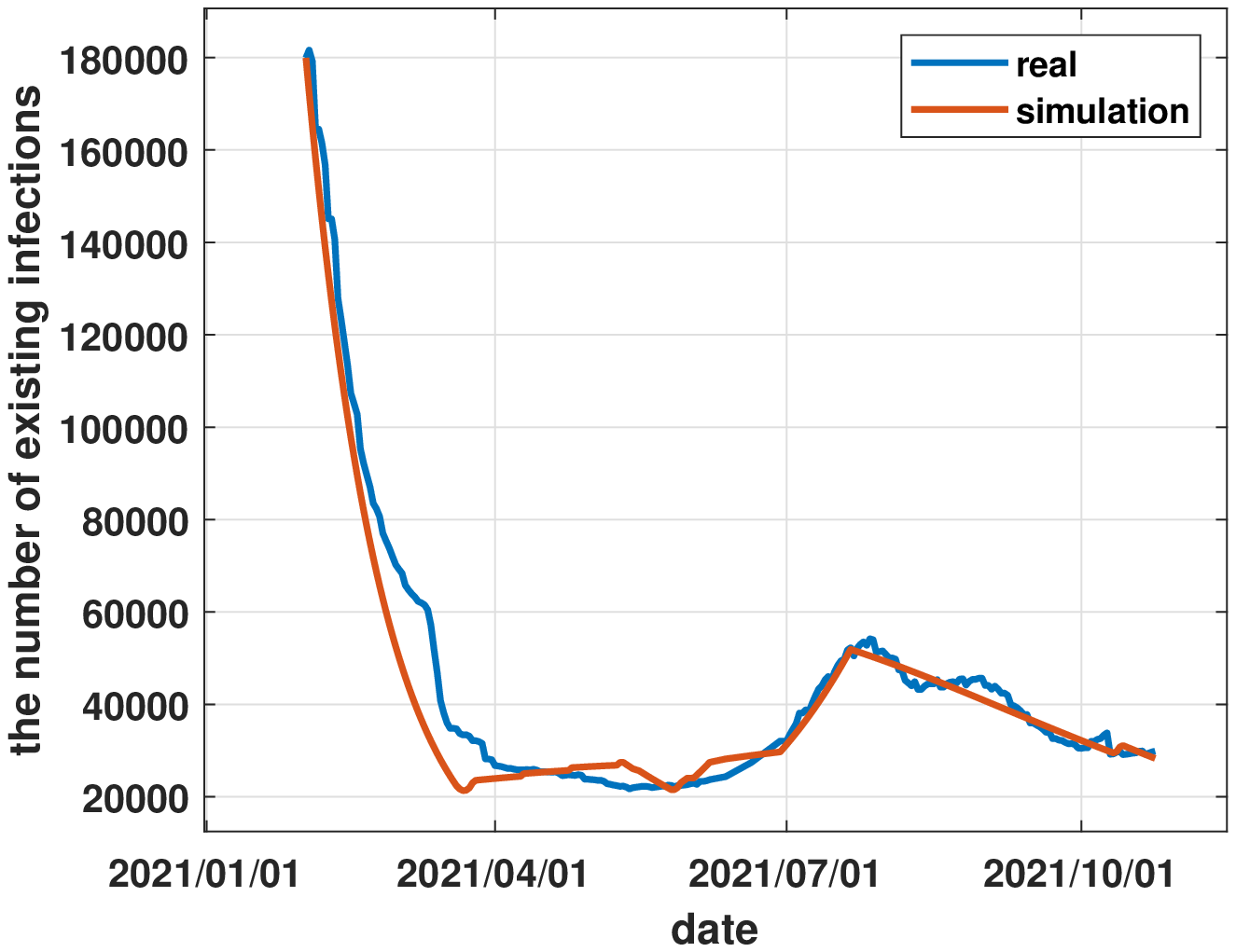}}
	\subfloat[\scriptsize Libya ($D^2=0.9713$)]{\includegraphics[width=0.48\textwidth]{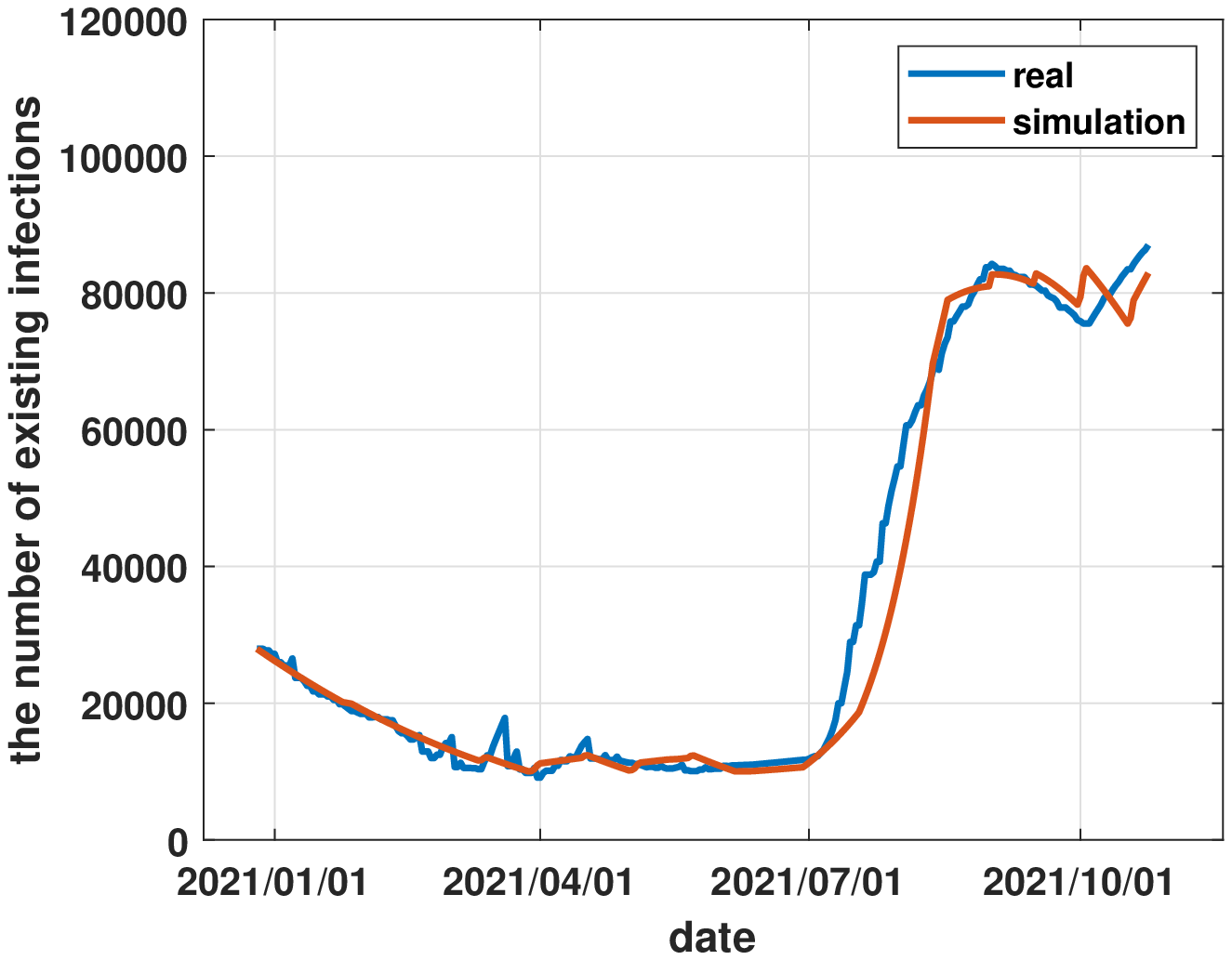}}
	\caption{\small Examples of countries that have a clear plateau phenomenon.
        The number of existing infections is calculated by the confirmed cases, minus the sum of recovery cases and death cases.
	}
	\vspace{-0.025in}
	\label{realworlddata}

\end{figure*}

Figure~\ref{realworlddata} shows 4 countries that have experienced plateaus, including Poland, the Philippines, Portugal, and Libya:

\begin{itemize}
    \item Poland, with a population of 38 million, had 16,478 existing infections at the beginning of October 2020. The number of existing infections had risen for 50 days, reaching 440,000. Then, possibly under the prevention of the local government, the number of existing infections dropped to about 230,000 and entered a brief plateau. However, in March 2021, the number of existing infections in Poland rose again to 400,000. It was quickly brought under control in April and entered a three-month plateau in June with 150,000 existing infections. In October, the number of existing infections in Poland showed another upward trend.
    \item The Philippines, with a population of 102 million, had been on a plateau from November 2020 to March 2021, and the number of existing infections remains at around 30,000. After that, the epidemic had rebounded, with the number of existing infections reaching as high as 190,000 in early April 2021. Under a period of prevention, the number of existing infections in the Philippines dropped to 50,000 and entered a new plateau (May-July). In August 2021, the epidemic in the Philippines rebounded again, and the number of existing infections rebounded to 180000 before being brought back under control in late September. The Philippines may enter a plateau again after that.
    \item Portugal, with a population of 10.35 million, had about 180,000 existing infections in February 2021. It was possible that under the adaptive prevention of the local government, the number of existing infections had been controlled at 20,000 to 30,000 since April, entering a three-month plateau. In July, the epidemic in Portugal intensified again. Fortunately, when the number of existing infections rose to 54,000, it was soon brought under control. Until October, the decline in the number of existing infections in Portugal slowed down and may once again enter a new plateau.
    \item Libya, with a population of 6.87 million, had 28,000 existing infections at the end of December 2020. The number of existing infections had been declining until it reached a plateau in April 2021 and remained at around 10,000. The plateau lasted three months and the outbreak broke out again in July. The number of existing infections in Libya rose to 80,000 in August and it would rise further in October. But if the Libyan government takes appropriate prevention, the epidemic will enter a new plateau and keep the number of existing infections around 80,000 without rising.
\end{itemize}

These data clearly support the existence of the plateau phenomenon in the real world, which means it is important for us to study in more detail the plateau phenomenon. Although it may be difficult to pinpoint the exact reasons that cause the plateau phenomenon for each country, from our simulation result observations, we believe that the plateau phenomenon in these countries is likely related to the adaptive prevention strategies adopted by the countries.

For most countries, especially those where the COVID-19 epidemic is still worsening, if proper adaptive prevention strategies are enforced, we could expect that they will also enter the plateau in the future. Thus, understanding the plateau phenomenon would be useful for most countries that are suffering from the COVID-19 epidemic. 

\subsection{Theoretical Results}
In the simulation study, we find that epidemic development might enter a plateau period after the SIR model is integrated with an adaptive prevention strategy. The results suggest some simple relationships such as the emergence of the plateau is related to the changing range of prevention parameter $[r_L,r_H]$, and the height of the plateau is positively correlated with the prevention threshold $\theta$. In this section, we will analyze theoretically the conditions of the emergence and the ending of the plateau as well as the properties of the plateau. Our analysis focus on the single-city model, while the multi-city model involves much more complicated dynamics due to the mobile populations between the cities.
	
\subsubsection{Effective Prevention Strategy and the Ideal Model}
As suggested from the simulation study, the plateau is related to the adaptive prevention strategy, but it is also the case that not every prevention strategy would generate plateau  (See Figure~\ref{changingrange}). We infer that this is related to the effective prevention parameter $\tilde{r}_t$, as defined in Supplementary Section A.
The effective prevention parameter $\tilde{r}_t$ is the value of prevention parameter $r$ when the effective reproduction number $R_t=1$.
Intuitively, when $R_t=1$, the number of infections an infected individual makes before recovered is exactly $1$, meaning that the number of existing infected cases would remain the same.
Thus, $\tilde{r}_t$ is the corresponding prevention parameter that would maintain the infection cases the same.
Note that in general the effective prevention parameter $\tilde{r}_t$ is different from the actual value of the prevention parameter $r$ used: the former is the would-be value to keep the infection cases the same while the latter is the actual value of the prevention strategy used at each day. This is also the reason we use different notations to distinguish the two. The effective prevention parameter $\tilde{r}_t$ has the following simple but important property:
\begin{lemma} \label{lem:monotoner}
	The effective prevention parameter $\tilde{r}_t$ is monotonically increasing with time $t$.
\end{lemma}
\begin{proof}
	$\tilde{r}_t$ is defined as the value of $r$ when $R_t=1$, we have
	\begin{align}
	\label{effectivereproductionnumber}
    R_t &= \frac{r\beta}{\gamma} \cdot \frac{S}{N}\\
	\label{eq:effectiveprevention}
	\tilde{r}_t &= \frac{N \cdot \gamma}{S \cdot \beta}.
	\end{align}
	Since the number of susceptible people $S$ is monotonically decreasing with time $t$ while other parameters $N,\beta,\gamma$ stay constant, we have  that $\tilde{r}_t$ is monotonically increasing.
\end{proof}

In the model, we have several aspects that try to reflect the real-world scenario, such as 
	the 14 day grace period before relaxing the prevention parameter, and 
	the changing of the strategies taking some time governed by parameter $k$.
To emphasize the essence of the model, for the theoretical analysis, we consider the following simplified {\em ideal} single-city model, which is consistent with the principle of full model .
In the ideal model, once the daily new cases drop to the prevention threshold $\theta$, we immediately adjust the prevention strategy
	$r$ such that we keep the daily new cases to be exactly $\theta$, as long as we have $r\le r_H$.
Moreover, we assume that the city starts with a significant number of infected cases and under intense prevention control, 
	and thus the prevention strategy is relaxed when the infection cases drop below the prevention threshold.
Technically we require that the initial infection number $I_0$ satisfies $I_0 \cdot \gamma > \theta$, that is, the number of daily recovered people initially is $I_0 \cdot \gamma$ and it is more than $\theta$. This corresponds to the real-world scenario where the prevention strategies are only employed after the epidemic outbreak already begins, and the adaptive strategies are designed to maintain the infection cases at a low level after the initial infection peak already occurs. Note that the above description on the ideal model only applies to the adaptive prevention strategy. The constant prevention strategy remains the same in the ideal model. We will also use the simulation results from the full model to validate the accuracy of the prediction from the ideal model.

\subsubsection{The Emergence and the Ending of Plateau}

We first investigate the conditions for the presence or the absence of the plateau, in the single-city model.
Let $t_\theta$ be the time when the number of daily recovered people drops to $\theta$, i.e. $I\gamma = \theta$.
We assume that initially $I_0 \gamma > \theta$, and under intense prevention control, thus $I$ will eventually drop to $\theta$.
With Lemma~\ref{lem:monotoner}, we can derive the following result connecting the effective prevention parameter with the plateau phenomenon.
\begin{theorem}
	\label{condition}
	In the ideal single-city model with an adaptive prevention strategy, at any time $t$, if the effective prevention parameter at this time $\tilde{r}_t \ge r_H$, then there will be no plateau after time $t$. On the other hand, if $\tilde{r}_{t_\theta} < r_H$, then plateau will start at $t_\theta$.
\end{theorem}
\begin{proof}{(Sketch)}
When $\tilde{r}_t \ge r_H$, we know that for any time $t'>t$, $\tilde{r}_{t'} > \tilde{r}_t$ according to Lemma~\ref{lem:monotoner}.
This means that for all time $t'>t$, the actual prevention parameter $r$ as the property that $r \le r_H < \tilde{r}_{t'}$.
Since $\tilde{r}_{t'}$ corresponds to the effective reproduction number 
	$R_{t'} = 1$, this implies that $R_{t'} < 1$ by Equation~\eqref{effectivereproductionnumber}, which in turn implies that
	for all time $t' > t$, the number of new infections any infected people would generate is less than $1$, and the epidemic will keep receding.

Now suppose that $\tilde{r}_{t_\theta} < r_H$.
According to the definition of $\tilde{r}_{t_\theta}$, $I \gamma = \theta$ at time $t_\theta$.
The ideal mode can be adjusted to $r=\tilde{r}_{t_\theta}$ at which the daily infection number balances with the daily recovery number $I \gamma$.
It can always keep the prevention parameter $r$ as $\tilde{r}_t$ until it increases to $r_H$.
During this period, the daily infection and daily recover balance each other, i.e. the plateau.
\end{proof}

The above theorem shows that the relationship of the effective prevention parameter $\tilde{r}_t$ and the upper bound of the prevention parameter range $r_H$ directly
	determine the presence of the plateau.
Note that for a constant prevention strategy with a fixed value $r$, it is clear that it has no plateau because it is essentially the
	same as the standard SIR model.
Therefore, we can obtain the following conclusion on the emergence and the ending of the plateau.
\begin{corollary}[Condition for the emergence of the plateau]
	\label{cor:emergenceplateau}
	In the ideal single-city model, the condition that the plateau will emerge is that we use an adaptive prevention strategy and
	$\tilde{r}_{t_\theta} < r_H$.
\end{corollary}
\begin{corollary}[Condition for the ending of the plateau]
	\label{cor:endplateau}
	In the ideal single-city model, given that plateau appears at time $t$, it will end at time $t'>t$ when $\tilde{r}_{t'} \ge r_H$.
\end{corollary}
Corollary~\ref{cor:emergenceplateau} states that plateau will appear when the effective prevention parameter is still less than $r_H$ and the daily recovery number hit the prevention threshold, while Corollary~\ref{cor:endplateau} states that the plateau will end once the effective prevent intensity reaches $r_H$.
We defer the discussion on the actual calculation of the condition $\tilde{r}_{t_\theta} < r_H$ to Section~\ref{sec:length} when we study the length of the plateau.

We use simulation results shown in Figure~\ref{Ircomparison} to validate these conditions in our actual model.
Figure~\ref{Istate} shows the epidemic development under two adaptive prevention strategies.
The first strategy, shown in the blue curve, has a clear plateau period.
In Figure~\ref{rstate}, we show the corresponding changes of the effective prevention parameter and the actual prevention parameter.
For the blue curve, Figure~\ref{rstate} shows the corresponding changes in the actual prevention parameter $r$, which fluctuates between $r_L=6$ and $r_H=10$ for a period of time.
During this period, the yellow curve, corresponding to the effective prevention parameter $\tilde{r}_{t,1}$ of the first strategy increases gradually from $8$ to $10$.
Since $\theta=40$ and $\gamma=0.12$, when y-axis $I$ in Figure~\ref{Istate} is $40/0.12 \approx 333$,
	the corresponding effective prevention parameter is around $8$, which is less than $r_H=10$, consistent with Corollary~\ref{cor:emergenceplateau} that the plateau would appear.
When $\tilde{r}_{t,1}$ increases to $10 = r_H$, we can see that the plateau ends in the corresponding blue curve on the left.
For the second strategy, the red curve in Figure~\ref{Istate} without plateau.
The actual prevention parameter $r$ keeps at the level $r_H=8$ in this case (the red curve in Figure~\ref{rstate}), while the effective 
	prevention parameter $\tilde{r}_{t,2}$ starts at $8$ and increases slowly above $8$ (the purple curve in Figure~\ref{rstate}).
This is also consistent with Corollary~\ref{cor:emergenceplateau}, meaning that when $\tilde{r}_{t_\theta} \ge r_H=8$, there will be no plateau in this case.

\begin{figure}[t]
	\centering
	\vspace{-0.4cm}
	\subfloat[\scriptsize The change curve of I]{\includegraphics[width=0.48\textwidth]{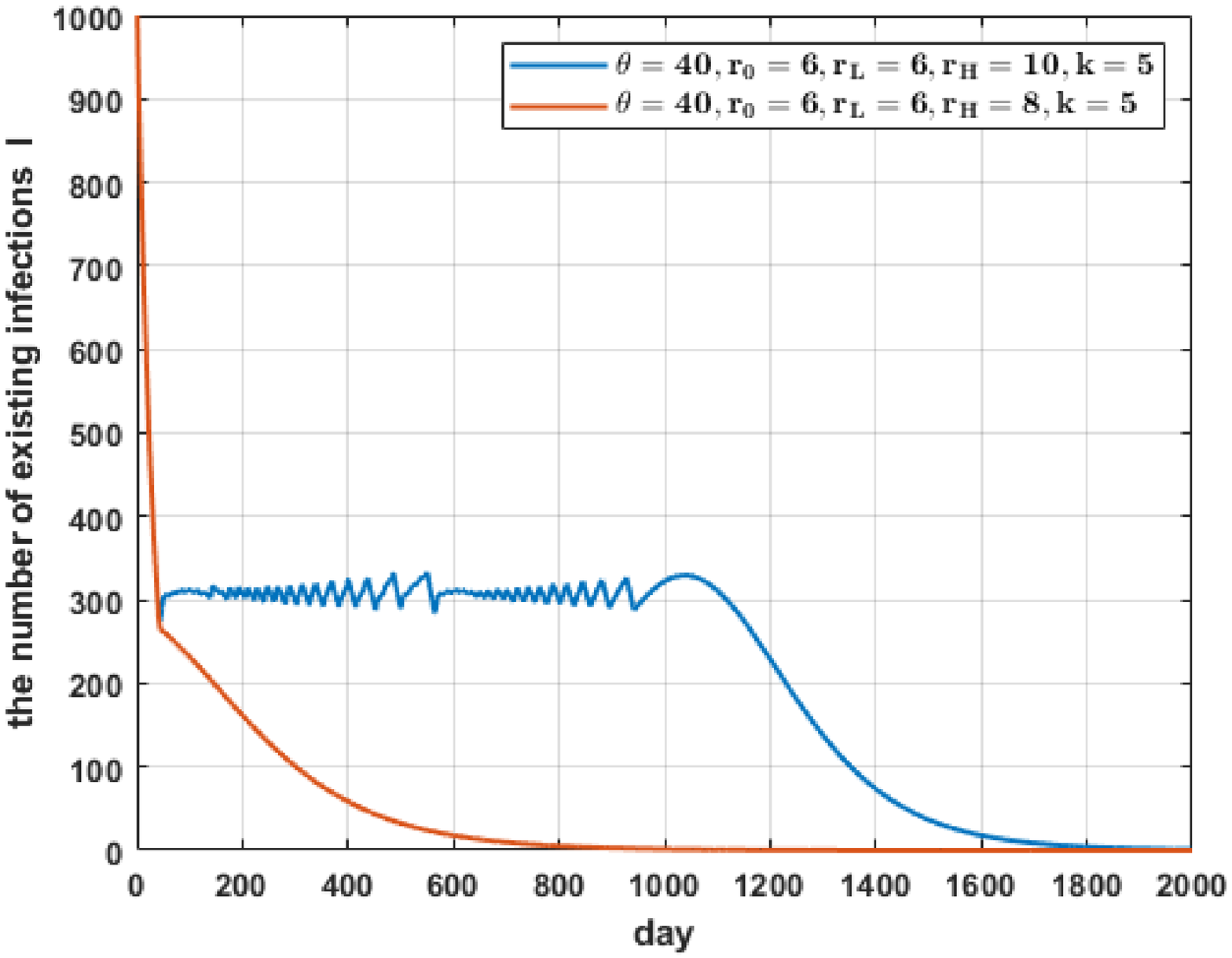}\label{Istate}}
	\subfloat[\scriptsize The change curve of r]{\includegraphics[width=0.48\textwidth]{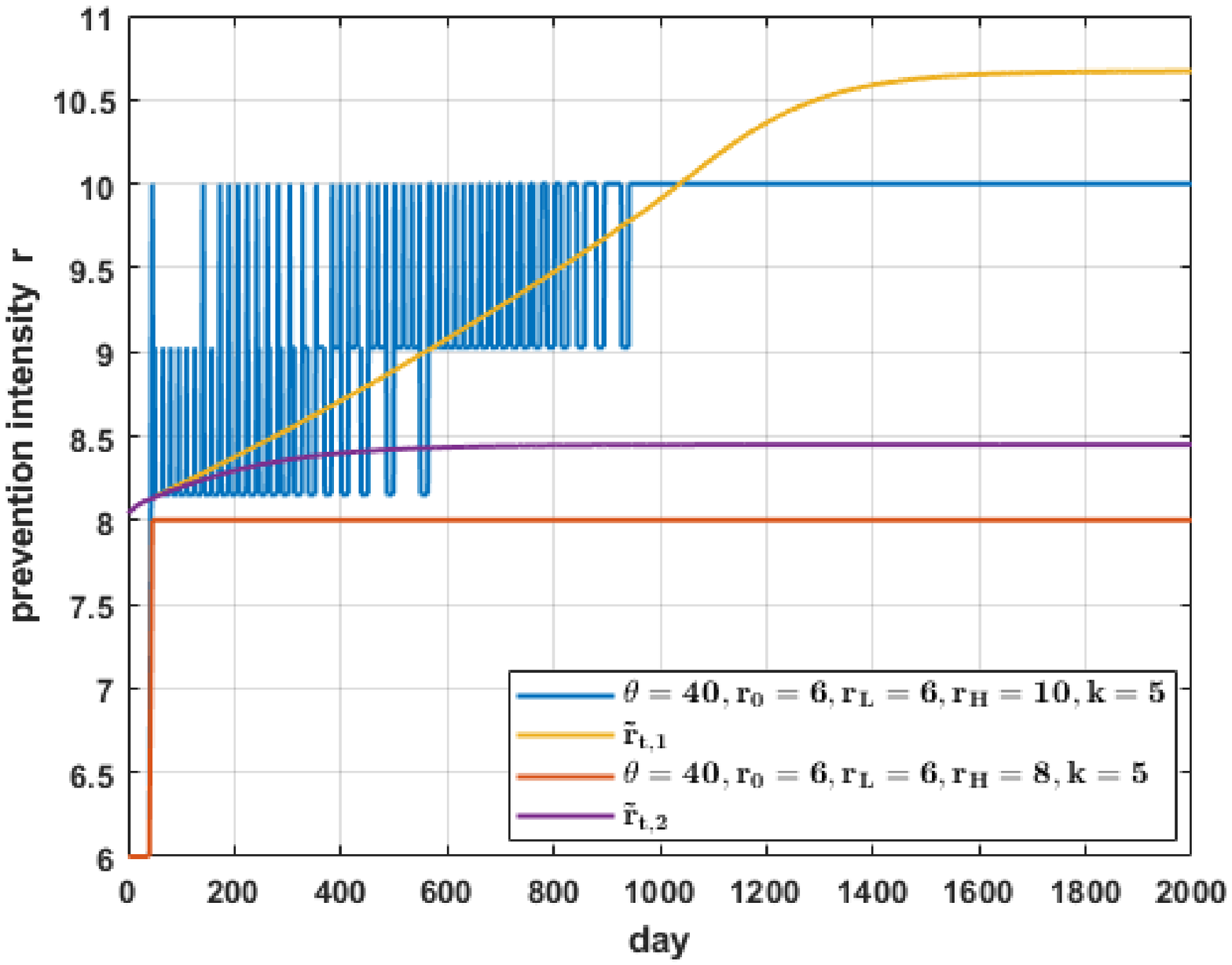}\label{rstate}}
	\caption{The relationship between the plateau and the prevention parameter. $\tilde{r}_{t,1}$ expressed the changes in the effective prevention parameter of the first prevention strategy, while $\tilde{r}_{t,2}$ stands for the second prevention strategy.
	}
	\label{Ircomparison}
	\vspace{-0.4cm}
\end{figure}

\subsubsection{The Height of the Plateau}

Although the number of infected people in the plateau has small fluctuations, it is basically maintained at a same value.
We call the value around which the number of existing infections fluctuates is the height of the plateau, and denote it as $I_P$.
In the ideal model, the height of the plateau $I_P$ is when the number of existing infected people does not change, that is, when the
	effective reproduction number $R_t=1$.
Using the relations we derived before, we can obtain the following result on $I_P$.
\begin{theorem} \label{thm:height}
In the ideal single-city model, the height of the plateau is given as follows:
\begin{equation*}
 I_{P}=\frac{\theta}{\gamma}, 
\end{equation*}
where $\theta$ is the prevention threshold and $\gamma$ is the recovery rate.
\end{theorem}
\begin{proof}
In the plateau we have $R_t=1$, and the newly infected people per day is exactly equal to the prevention threshold $\theta$ (in the ideal model).
According to the SIR model, the newly infected people per day is $\frac{r \beta IS}{N}$ (excluding the number of recovered people).
Thus in the plateau,
\begin{align}
\theta&=\frac{r\beta I_P S}{N}. \label{eq:theta-di}
\end{align}
From Equation~\eqref{effectivereproductionnumber}, we know that $\frac{r\beta S}{N} = \gamma$.
Plugging it into Equation\eqref{eq:theta-di}, we have $I_P = \theta / \gamma$.
\end{proof}

Theorem~\ref{thm:height} directly shows that the height of the plateau has positive linear relationship with the prevention threshold $\theta$, and it is not
	related to the other parameters such as the range of prevention strategy $[r_L, r_H]$ or the speed of prevention $k$.
This result is consistent with our simulation study findings.
In particular, comparing with Figure~\ref{threshold}, 
	when we plug in $\theta = 10, 20, 40, 80$ with $\gamma = 0.12$, we obtain the theoretical heights as 
	$I_P = 83, 167, 333, 667$, while the experimental averages we obtain the simulation results
	are $77,154,308,622$ respectively, indicating that the theoretical
	results match with the simulation results reasonably well.
The reason that the simulation results are consistently lower than the theoretical predictions is because in our actual model, we have an 
	asymmetric strategy adjustment: we need 14 consecutive days of new cases dropping below $\theta$ to relax the prevention, while we only need
	one day of new cases rising above $\theta$ to strengthen the prevention.
This asymmetric control push the actual height lower than our theoretical prediction from the ideal model.

%
%
%
%

\subsubsection{The Length of the Plateau}
\label{sec:length}

After obtaining the height of the plateau, we can further calculate the length of the plateau by the Area Method.
We denote the length of the plateau, i.e. the time period from when the plateau begins to the plateau ends, as $T$.
It is clear that $T \cdot I_P$ is the area covered of the plateau.
We just need to calculate the area underlying the plateau to get $T$.
The result is given in the following theorem.

\begin{theorem} \label{thm:length}
In the ideal single-city model, assuming that the number of new infections falls to $\theta$ is the beginning of the plateau, the length of the plateau is given as follows:
\begin{equation}
 T=\frac{S_B - S_E}{\theta}, \label{eq:length}
\end{equation}
where $S_E$ is the number of susceptible people when the plateau ends, and $S_E=\frac{\gamma N}{\beta r_H}$; $S_B$ is the number of susceptible people when the plateau begins, and it is the solution of the following equality:
\begin{align}
S=(S_0+I_0)- I_P +\frac{N\gamma}{r_0\beta}\ln{\frac{S}{S_0}}, \label{eq:I-S}
\end{align}
where $S_0, I_0$ are the number of initial susceptible and infected people, and $r_0$ is the initial prevention parameter in the range $[r_L, r_H]$.
\end{theorem}

\begin{proof}{(Sketch)}
In the curve with $I$ as y-axis and time as x-axis, the area under curve (AUC)  by time $t$ is the cumulative existing infected people 
	 by time $t$.
On average, each infected individual takes $1/\gamma$ time to recover, so each infected individual is counted $1/\gamma$ times.
Therefore, the AUC by time $t$ is the total number of once infected people multiplied by $1/\gamma$, which is $(I+R) /\gamma$.
Since $S+I+R = N$, so we have the AUC by time $t$ is $(N-S)/\gamma$.
Let $t_B, t_E$ be the time when the plateau begins and ends, and correspondingly $S_B$ and $S_E$ are the numbers of susceptible people
	at time $t_B$ and $t_E$ respectively.
Then we have the AUC between $t_B$ and $t_E$ is $(N-S_E)/\gamma - (N-S_B)/\gamma = (S_B - S_E) / \gamma$.
Since the height of the plateau is $\theta/\gamma$ by Theorem~\ref{thm:height}, we know the length of the plateau
	$T = \frac{(S_B - S_E) / \gamma}{\theta/\gamma} = \frac{S_B - S_E}{\theta}$.
Thus Equation\eqref{eq:length} holds.

We now show how to derive $S_B$ and $S_E$.
First, for $S_E$, we know from Corollary~\ref{cor:endplateau} that when the plateau ends, $\tilde{r}_{t_E} = r_H$, and by the definition
	of $\tilde{r}_{t_E}$, we know $R_{t_E} = 1$.
 Plugging these values into Equation~\eqref{effectivereproductionnumber}, we obtain that $S_E = \frac{\gamma N}{\beta r_H}$.

To derive the result for $S_B$, we rely on the equation below connecting $I$ and $S$ in the SIR model~\cite{Zhen2004}:
\begin{align*}
I=(S_0+I_0)-S+\frac{N\gamma}{r\beta}\ln{\frac{S}{S_0}}, 
\end{align*}
Before the plateau appears, the prevention parameter $r$ is fixed to a value $r_0$.
When the plateau appears, $I$ is equal to $I_P$.
Thus, when plugging $I=I_P$ and $r=r_0$ into the above equation, we obtain Equation\eqref{eq:I-S}, and $S_B$ is the solution of this equation.

\end{proof}

Note that $S_B$ derived from Equation~\eqref{eq:I-S} can be used to test condition $\tilde{r}_{t_\theta} < r_H$ in 
	Corollary~\ref{cor:emergenceplateau}:  $\tilde{r}_{t_\theta} = \frac{N \cdot \gamma}{S_B \cdot \beta} $ according to Equation\eqref{eq:effectiveprevention}.
	
We now validate Theorem~\ref{thm:length} with the experimental results.
We compare our theoretical prediction on the plateau length with the experimental results from Figure~\ref{threshold}, 
	with the initial condition $IC = (200000,0.5\%)$ and parameters $r_0=6, r_L=6, r_H=10, k=5$, $\theta$ from $10$ to $80$.
Table~\ref{tab:height} summarizes the comparison result.
We can see that the theoretical and experimental results match reasonably well.
Moreover, the estimation error is mainly from the error in the estimation of the height --- if we replace the theoretical heights with
	the experimental heights, we get much closer estimates of $3920, 1984, 1002, 509$ respectively, which means our theoretical AUC prediction
	is quite accurate.

\begin{table*}[t]
    \setlength{\abovecaptionskip}{2mm}
	\centering
    \caption{\small The experimental and theoretical heights and lengths of the plateau for settings in Figure~\ref{threshold}.}
	\small
	\label{tab:height}
	\setlength{\tabcolsep}{10mm}
	\begin{threeparttable}				
		\begin{tabular}{l|c|c|c|c}
			\toprule
			\small $\theta$ &  \small 10  &  \small 20  &  \small 40 &  \small 80  \\
			\midrule 
			I (experimental)        & 77 & 154 & 308 & 622 \\ 	
			\midrule 
			I (theoretical)        & 83 & 167 & 333 & 667 \\
			\midrule 
			days (experimental)        & 3,929 & 1,975 & 1,002 & 514 \\ 	
			\midrule 
			days (theoretical)        & 3,637 & 1,830 & 927 & 475 \\ 	
			\bottomrule
		\end{tabular}
	\end{threeparttable}
\end{table*}

\section{Conclusion and Discussion}
In this paper, we study an epidemic development phenomenon: plateau. This phenomenon is formed under the change of prevention strength of adaptive prevention strategies.  In the traditional infectious disease model, it is difficult to find the existence of the plateau.  We investigate empirically from both simulations and real-world data on the plateau phenomenon, in both single-city model and multi-city model, and further provide theoretical analysis in the single-city model. In the two years since the outbreak of COVID-19, we found more and more countries have entered a plateau. It is difficult to achieve a complete regression of the epidemic in a short term. It may be a better solution to control the epidemic through adaptive prevention strategies. In the single-city model, we obtain a number of results on the emergence and the characteristics of the plateau phenomenon. Our results in the multi-city model suggest that all regions are interconnected and local prevention strategies would have effects on other regions as well. The plateau phenomenon well explains the persistence of the current COVID-19 epidemic and provides a new solution for countries still suffering from the epidemic.

There could be some future work along this direction. On the empirical side, one could investigate in more detail the real-world plateau phenomenon and their correlations with specific prevention strategies used. On the theoretical side, the multi-city model is much more challenging due to the complicated population mobility dynamics. It would be very important and beneficial to obtain some insightful results in the multi-city model. We hope that our findings could help the understanding of the COVID-19 pandemic (and other epidemics as well) in connection with different prevention strategies, and assist policymakers in predicting the effect of their prevention strategies. At present, our findings are being applied to the epidemic prevention work in Shenzhen.

%
%

\section{Materials and Methods}
\subsection{Experimental Hypothesis}
We base the choice of our parameters on the research that has been made on the COVID-19.

A person stays infected by COVID-19 for 8.4 days on average \cite{wu2020nowcasting}. The recovery rate is thus
defined as $\gamma= 1/8.4 = 0.12$.
For the infection rate $\beta=0.015$, we use the most common indicator in the epidemic research, the reproduction number $R$ to calculate, as shown in Supplementary Section A in detail.

We used the epidemic data provided by the Information Center of Shenzhen Public Security Bureau to better estimate the number of contacts --- prevention parameter $r$.
We estimate the number of contacts among people $r=24$ under normal conditions.
With the assistance of the Information Center of Shenzhen Public Security Bureau, the population mobility ratio $p$ is estimated according
to the population mobility data of Shenzhen in 2019 and 2020.
During the outbreak of COVID-19 in 2020, Shenzhen implemented strict restrictions on population mobility, and the proportion of population mobility was around $5\%$.
In the same period in 2019, the proportion of people moving fluctuated between $10\%$ and $20\%$.
In this paper, it is assumed that the proportion of population mobility under normal conditions is $p=20\%$.

\subsection{Epidemic Model with Prevention Strategies}
We incorporate the constantly changing prevention measures as well as the population mobility between regions into the SIR model.
For the COVID-19 with a latent period, we do not choose the SEIR model, because using the SIR model in the process of studying the prevention strategy can simplify the analysis process, and is no different from the SEIR model.

To incorporate the prevention strategies and their dynamic changes into the model, we add a prevention parameter $r$, which represents the number of people that individuals will contact per unit of time.
Cities dynamically adapt their prevention measures based on the current infection data --- when the new infection cases
	exceed expectations the prevention measures intensify, and when the number of new cases drops, the prevention measures are relaxed.
We model this scenario by setting a prevention threshold and dynamically changing the prevention parameter $r_i$ and 
	the fraction of the floating population $p_i$
	based on whether the current infection cases are below or above the prevention threshold.
Technically, we have three main factors for deciding and changing the prevention measures: prevention threshold, the range of changing prevention parameter, and the reaction speed.
Based on the variability of parameters, it is divided into constant prevention strategy and adaptive prevention strategy.

We extend the single-city model to multiple cities that $\beta$ and $\gamma$ are the same across cities since they model the intrinsic characteristics of the virus.
The prevention strategies are different across cities, modeling different prevention measures taken among the cities.
In order to ensure the dynamic balance of population mobility in multiple cities, we put forward a reasonable rule of population mobility and combine it with the external environment $EE$, so that population mobility can be added to the model.
Model details and code can be found in Supplementary Section B.

\bibliographystyle{unsrt}  

\bibliography{references}




\newpage

\appendix{
\section*{Supplementary Materials}
\section*{Section A: Hypothetical Basis}
\label{setting}
\subsection*{Infection Rate and Recovery Rate}
A person stays infected by COVID-19 for 8.4 days on average \cite{wu2020nowcasting}. The recovery rate is thus
defined as $\gamma= 1/8.4 = 0.12$.
For the infection rate $\beta$, we use the most common indicator in the epidemic research, the reproduction number $R$ to calculate.

The reproduction number $R$ can be divided into basic reproduction number and effective reproduction number.
The basic reproduction number represents the average number of people infected in a group of all susceptible people before an infected person recovers, expressed as $R_0$.
The estimation of the basic reproduction number is generally chosen in the early stage of the outbreak, as the proportion of susceptible people in the population is close to $1$.
The effective reproduction number is different from the basic reproduction number, and it
	represents the average number of people infected at time $t$ with a certain proportion of the susceptible population before an infected person recovers, expressed as $R_t$.
The effective reproductive number $R_t$ will change over time, because the proportion of susceptible people will change with the development of the epidemic.
When the effective reproduction number $R_t<1$, it indicates that the epidemic is receding.
The time from infection to the recovery is an infection cycle, which is $\frac{1}{\gamma}$ days.
Before recovering, the infected person will come into contact with $r$ individuals every day.
Among these people, only susceptible people will be infected, and the probability of each susceptible individual
	getting infected is $\beta$.
Therefore, the basic reproductive number $R_0$ and the effective reproductive number $R_t$ have the relationship with infection rate $\beta$, recovery rate $\gamma$ and prevention parameter $r$ as shown in Equation~\eqref{basicreproductionnumber} and Equation~\eqref{effectivereproductionnumber}:

\begin{align}
\label{basicreproductionnumber}
R_0 &= \frac{r\beta}{\gamma}.
\end{align}
We can get the infection rate $\beta$ of COVID-19 through the basic reproduction number $R_0$ of SARS-CoV-2 and the prevention parameter $r$ as in Equation~\eqref{basicreproductionnumber}.

Before calculating the infection rate $\beta$, it is necessary to know the prevention parameter $r$, or the 
	number of human contacts per day.
As discussed before $r$ will change according to the prevention measures.
The basic reproduction number $R_0$ should corresponds to the situation of the early stage of the outbreak before prevention.
At this time, the proportion of susceptible people is close to 1, and the number of contacts with the population is at the normal level.
In order to better estimate the number of human contacts, we used the epidemic data provided by the Information Center of Shenzhen Public Security Bureau.
In Shenzhen, for people infected with COVID-19, the Public Security Bureau will track down the people who have been in contact with him during the incubation period and conduct medical observation to avoid further spreading.
The Public Security Bureau records the number of contacts tracked every day and refers to the number of people obtained as the number of medical observers.
Based on the number of daily medical observations, divided by the product of the number of new diagnoses per day and the incubation period ($3.0$ days), we estimate the normal number of contacts among people $r=24$.
In order to make the results closer to the normal number of human contact, the data from the early stage of the epidemic (the first 5 days of the outbreak in Shenzhen) were selected in the experiment.

There are some estimates of the basic reproduction number $R_0$ in the articles about COVID-19, and the results are all around 3~\cite{wu2020nowcasting, hao2020reconstruction, zhou2020preliminary, cao2020incorporating, de2020epidemiological}.
Assuming that the basic reproduction number of SARS-CoV-2 $R_0=3$, combined with the recovery rate $\gamma=0.12$ and the prevention parameter $r=24$, the infection rate of COVID-19 can be obtained as $\beta=0.015$ from Equation~\eqref{basicreproductionnumber}.

\subsection*{The Number of Population Contacts and the Proportion of Population Mobility}
\label{randp}
As derived above, $r=24$ is the number of contacts without prevention.
The prevention strategies would change $r$.
Only when the effective reproduction number $R_t < 1$, the epidemic begins to subside, and the number of infections declines.
We call the number $r$ that corresponds to $R_t=1$ as $\tilde{r}_t$, the effective prevention parameter.
Based on Equation~\eqref{effectivereproductionnumber} and the 
	calculated infection rate and recovery rate from above, we can obtain that 
	the number of infections will decrease immediately when the prevention parameter $r<8$.

With the assistance of the Information Center of Shenzhen Public Security Bureau, the population mobility ratio $p$ is estimated according
to the population mobility data of Shenzhen in 2019 and 2020.
During the outbreak of COVID-19 in 2020, Shenzhen implemented strict restrictions on population mobility, and the proportion of population mobility was around $5\%$.
In the same period in 2019, the proportion of people moving fluctuated between $10\%$ and $20\%$.
In this paper, it is assumed that the proportion of population mobility under normal conditions is $p=20\%$.

\subsection*{Initial Condition}
In the simulation experiment, the total population and the proportion of initial infected people in the city are called initial conditions.
Different initial conditions will appear in the single-city model and the multi-city model.
An initial condition can be expressed as $IC = (N,\frac{I_0}{N})$, where $N$ represents the urban population and $I_0$ represents the initial number of people infected.

\section*{Section B: Model Details}
\label{modeldetails}
Traditional epidemic models such as the SIR model do not model behavioral changes in the population that could change the virus propagation dynamics.
However, as we have seen with the COVID-19 epidemic, 
	people change their behaviors due to the epidemic either voluntarily or by government restrictions, such as wearing masks, reduce traveling and gathering when the epidemic is getting severe, until the epidemic situation
	is getting better.
Moreover, different geographic regions could impose different prevention strategies, and it is further complicated by the cross-regional population movement.

In this paper, we incorporate the constantly changing prevention measures as well as the population mobility between regions into the SIR model~\footnote{Code is available at https://github.com/Complex-data/Plateau}.
For the COVID-19 with latent period, we do not choose the SEIR model, because using the SIR model in the process of studying the prevention strategy can simplify the analysis process, and is no different from the SEIR model.
For simplicity, we refer regions as cities, but in reality they could represent either cities or provinces, states, counties, or communities.

\subsection*{SIR-based Single-city Model}

Our single-city model is based on the SIR model~\cite{newman10}, in which each individual goes through the states of
	being susceptible (S), infected (I), and recovered(R).
The standard SIR model is governed by two parameters --- infection rate $\beta$ and recovery rate $\gamma$.
To incorporate the prevention strategies and its dynamic changes into the model, we further add a prevention parameter $r$ into the model, which represents the number of people that individuals will be exposed to on time unit.
Thus, in our model $\beta$ is interpreted as the probability that when a susceptible individual interacts with 
	an already infected individual in a time unit, the susceptible individual gets infected from the latter,
	and $\gamma$ is the probability that an infected individual recovers from the disease in a time unit, while
	$r$ is the number of individuals one infected individual would interact with in one time unit.

It is important to remark that in our model, $\beta$ and $\gamma$ are the intrinsic properties of the virus, and it is not affected by
	the prevention measures, while $r$ is affected by the prevention measures --- $r$ will be smaller if the city imposes higher level
	of prevention measures and people exercise social distancing, and $r$ is larger if the city relaxes its prevention measures.
The combined $\frac{r\beta}{N}$ would correspond to the $\beta$ parameter in the standard SIR model.
Then, the differential equations of the single-city model are as follows:

\begin{align}
\frac{\mathrm{d}S}{\mathrm{d}t}& =-\frac{r \beta IS}{N},  \notag\\
\frac{\mathrm{d}I}{\mathrm{d}t}& =\frac{r \beta IS}{N}- \gamma I, \label{eq:single} \\
\frac{\mathrm{d}R}{\mathrm{d}t}& =\gamma I,  \notag
\end{align}
where $S$ is the susceptible population, $I$ is the infected population, $R$ is the recovered population, $N$ is the total population of the city, and $t$ is the time. 
Note that $S, I, R$ are considered variables that change over time, while $N$ remains the same over time.
The change of $S$ in a time unit $\mathrm{d}t$ due to the fact that 
	each infected individual meets $r$ individuals, among them a fraction of $S/N$ are susceptibles, and thus $\beta$ fraction of them
	will be infected and change from $S$ state to $I$ state, and since we have $I$ infected individuals, the total decrease on $S$
	is $\frac{r \beta IS}{N}$.
The change of $R$ is simply because there are currently $I$ infected individuals and $\gamma$ fraction of them would be recovered
	in each time unit $\mathrm{d}t$, and the change of $I$ is the combined effect of $S$ and $R$.

\begin{figure*}[!tpb]
	\centering
	\subfloat[\scriptsize]{\includegraphics[width=0.3\textwidth]{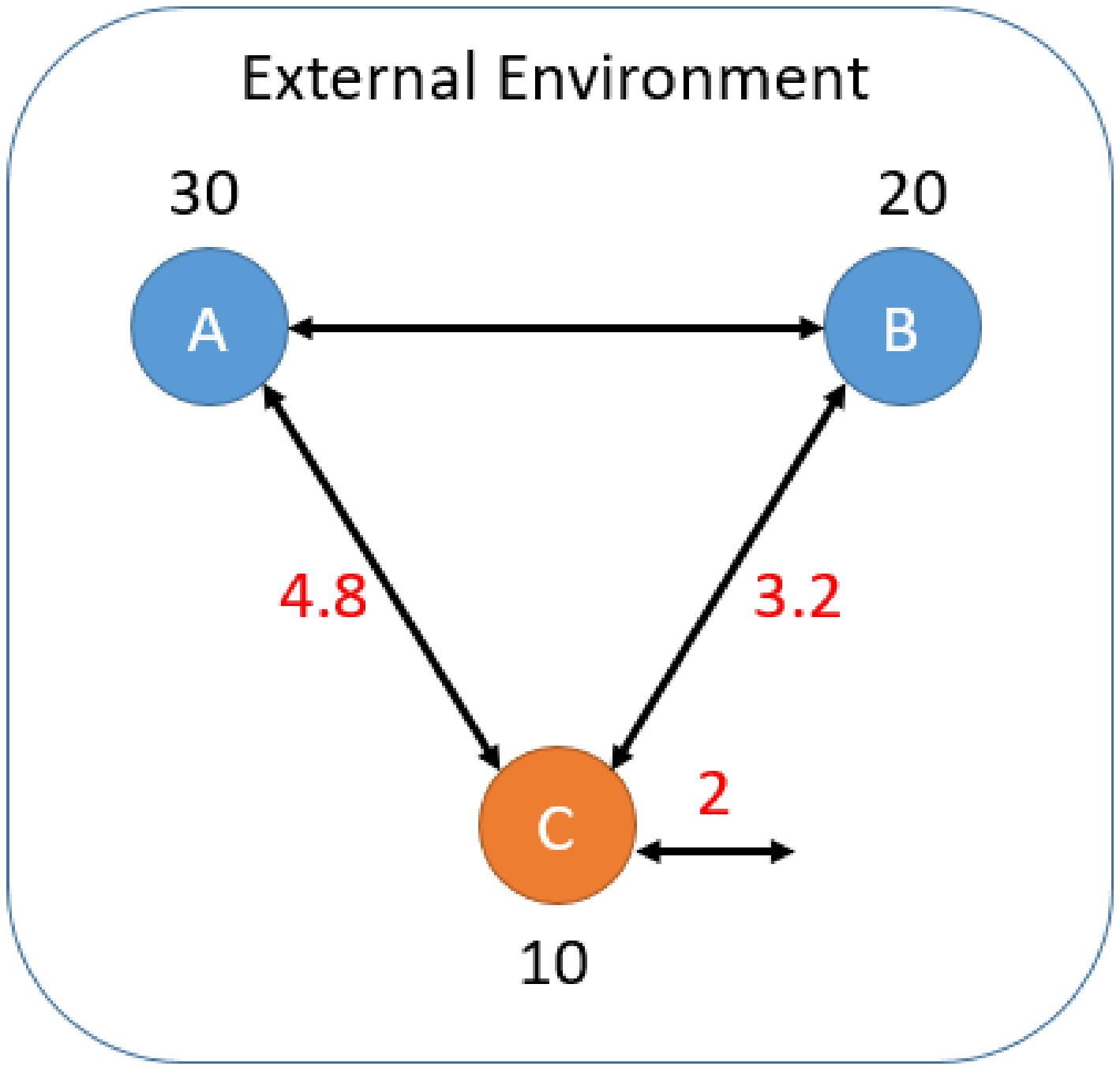}\label{pa}}
	\subfloat[\scriptsize]{\includegraphics[width=0.3\textwidth]{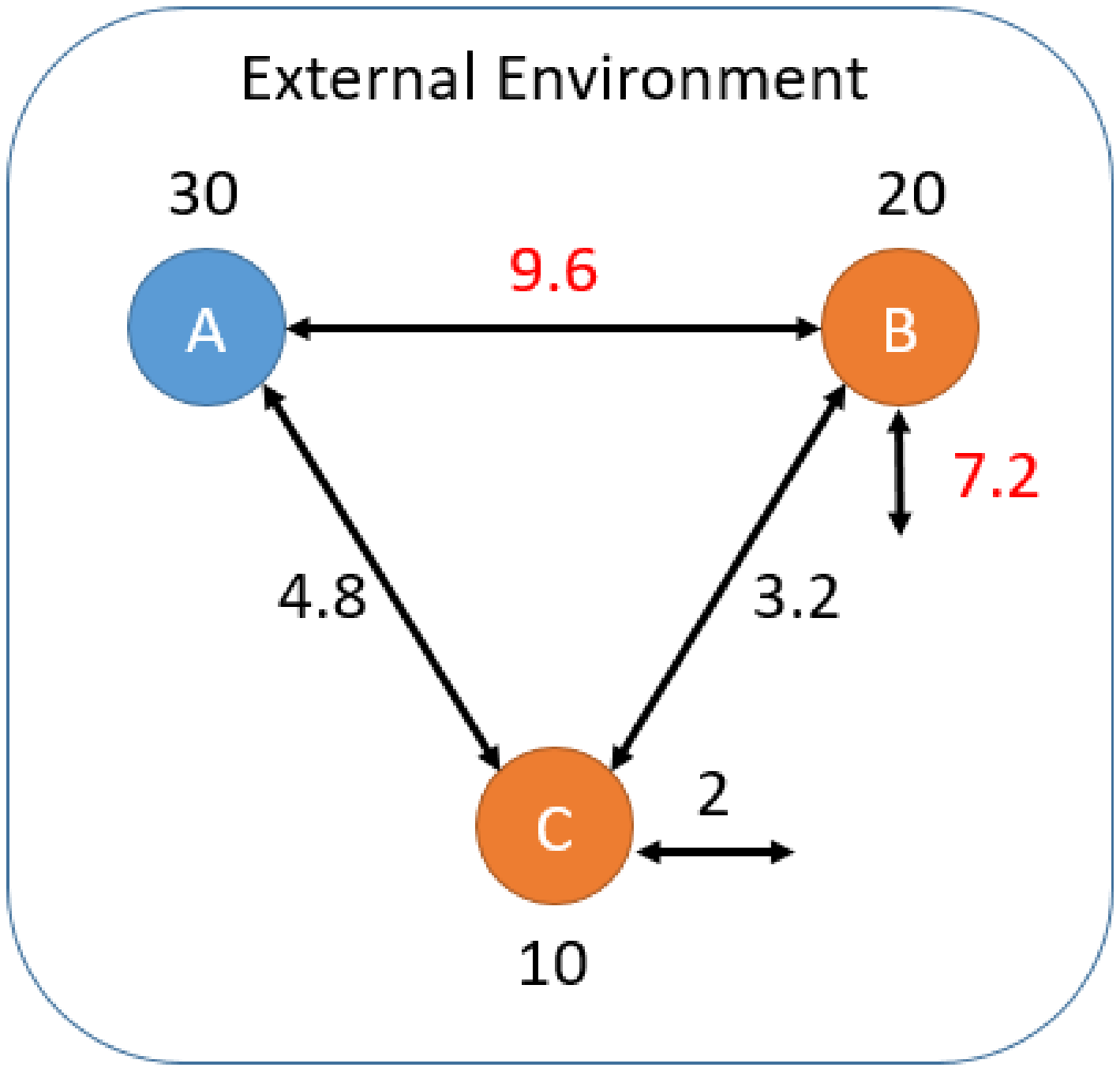}\label{pb}}
	\subfloat[\scriptsize]{\includegraphics[width=0.3\textwidth]{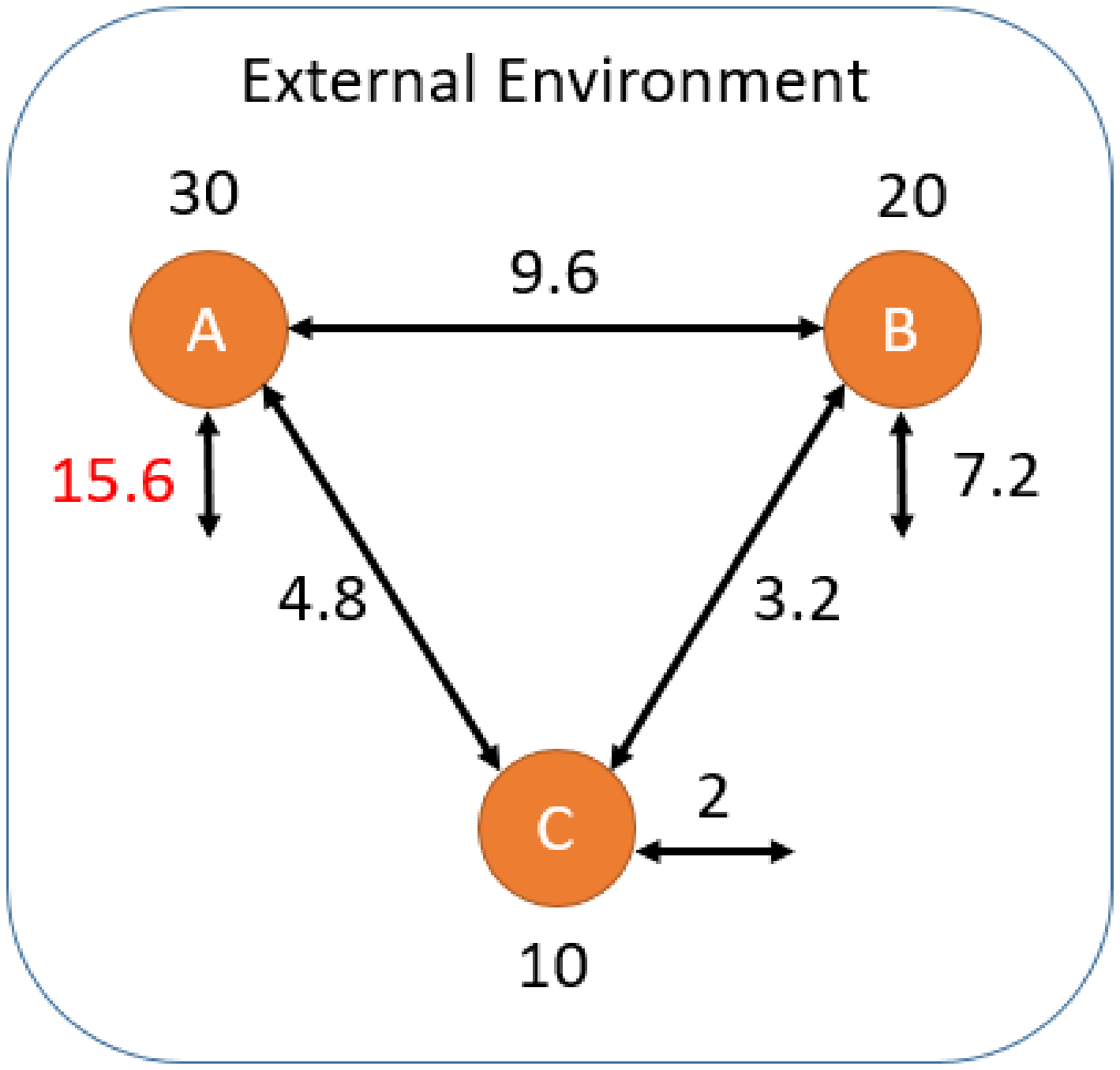}\label{pc}}
	
	\caption{\small An example of the mobile population rules.
	}
	\label{populationmobilityrule}
	\vspace{-0.025in}

\end{figure*}

\subsection*{Multi-city Model}

We extend the single-city model to multiple cities where different cities may impose different levels of prevention measures and there
	are population movement between the cities.
In the multi-city model, assume that there are $m$ cities.
For each city $i$, we use subscript $i$ in $S_i, I_i, R_i, N_i$ and $r_i$ to represent the SIR variables and parameters.
Note that $\beta$ and $\gamma$ are the same across cities since they model the intrinsic characteristics of the virus, while 
	the prevention parameter \{$r_i$\} are different across cities, modeling different prevention measures taken among the cities.

\noindent \textbf{Population balancing with the external environment}.
In the real world, the population base of a city is usually stable over time(inflows equals outflows).
Thus, we would like our model to maintain the population for city $i$, $N_i$, the same over time, which also simplifies our model
	analysis.
However, with different population mobility between cities,
	it is not trivial to maintain population balance in the multi-city model.
To achieve this, we introduce the external environment as the background that could counter balance the excess or deficit of population
	of each city due to the population movement.
 The external environment could be viewed as other cities that we would not model explicitly or the rural areas surrounding the cities.

As shown in Figure~\ref{populationmobilityrule}, suppose that each city $i$ has a proportion of the floating population $p_i$, then will has a floating population $fp_i$, which is the population moving in or out of
	the city in a time unit. 
Let $f_{i,j}$ be the population moving from city $i$ and city $j$ in a time unit, and $f_{i,0}$ be the population moving between the city $i$
	to the external environment in a time unit.
Thus, we have $fp_i = N_{i} * p_{i} = \sum\nolimits_{j=0,j \neq i}^{m} f_{i,j}$.

We use the following rules to govern the inter-city population movement:
(1) between two cities, the two directional movement population is balanced, i.e. $f_{i,j} = f_{j,i}$; and
(2) the population movements from one city $i$ to two different cities $j$ and $j'$ are proportional to the two destination cities'
	floating populations, i.e. $f_{i,j} / f_{i,j'} = fp_j / fp_{j'}$.
To realize the population movement, we use the following procedure:

\floatname{algorithm}{Procedure}
\begin{algorithm}
\caption{A round of population mobility between cities.}
\label{alg:c}
\hspace*{0.02in} {\bf Input:} \\
\hspace*{0.2in} $m$, the number of cities\\
\hspace*{0.2in} $N_i$, the population of city $i$\\
\hspace*{0.2in} $p_i$, the flow proportion of city $i$\\
\hspace*{0.2in} $p_{em}$, the minimum proportion of external floating population\\
\hspace*{0.02in} {\bf Output:} \\
\hspace*{0.2in} $f_{i,j}$, the floating population between city $i$ and city $j$
\begin{algorithmic}[1]
\State \textbf{function} Population\_Mobility 
\For{$i \gets 1:m$}
\State $fp_{i}=N_i \times p_i$
\State $fp_{sum}+=fp_{i}$
\EndFor
\State Ascending\_Sort$(fp_i)$
\State $f_{0,1}=p_{em} \times fp_1$
\For{$j \gets 2:m$}
\State $f_{1,j}=(1-p_{em}) \times fp_1 \times \frac{fp_j}{fp_{sum}-fp_i}$
\EndFor
\For{$i \gets 2:m-1$}
\For{$j \gets i+1:m$}
\State $f_{i,j}=f_{i-1,i} \times \frac{fp_j}{fp_{i-1}}$
\EndFor
\State $f_{0,i}=fp_i$
\For{$j \gets 1:m$}
\State $f_{0,i}-=f_{i,j}$
\EndFor
\EndFor
\State \Return $f_{i,j}$
\end{algorithmic}
\end{algorithm}

\begin{itemize}

\item   We first calculate the floating population $fp_i$ of each city and the total floating population $fp_{sum}$. Then we sort the floating population $fp_i$ from the smallest to the largest, and determine the population movement of each city in this order (without loss of generality, assuming $fp_1\le fp_2 \le \cdots$).
\item For the city $1$ with the smallest floating population, we first allocate a fixed fraction (e.g. $20\%$) of the floating population
	as the exchange population with the external environment to ensure that every city has population exchange with the environment, which is called the minimum proportion of external floating population $p_{em}$. The remaining floating population is allocated to $f_{1,j}$ 
	proportional to the city $j$'s floating population $fp_j$, according to rule (2) above.
\item Then for each remaining city $i$ in the order, we first apply rule (1) above to fix the $f_{i,j} = f_{j,i}$ for all $j<i$, and then
	apply rule (2) on the remaining cities $j' > i$ to get $f_{i,j'}$ proportional to $fp_{j'}$. Finally for the left over population
	$fp_i - \sum\nolimits_{j=1,j \neq i}^{m} f_{i,j}$, we assign it to the exchange population with the environment $f_{i,0}$.
     
\end{itemize}
We take city A, B and C as examples for a brief introduction below.


Figure~\ref{populationmobilityrule} shows a concrete example of the procedure.
\text{(a)} The floating population of city A, city B and city C is 30, 20 and 10 thousand respectively. First of all, the floating population between city C and other cities is calculated, with $20\%$ floating population (i.e. 2,000) exchanges with the environment, and
	the remaining $80\%$, i.e. 8,000 people are allocated among moving populations to city A and city B.
Based on the floating population proportion, movement to city B accounts for $40\%$, which is 3,200, and movement to city A is 4,800.
\text{(b)} Then we calculate the city B with the second smallest floating population. The floating population between city B and city C is 3,200, and the ratio of floating population between city C and city A is 1: 3. Therefore, the proportion of floating population between city B and city C and between city B and city A is also 1:3. The floating population between city B and city A is 9,600, and the remaining 7,200 are between city B and the external environment.
\text{(c)} Similarly, the remaining 15,600 floating population in city A will be used as the floating population with the external environment.

With the above rules of population mobility, the population base of each city remains unchanged, and the cities with more floating population account for a larger proportion of the floating population in other cities.

\noindent \textbf{Multi-city SIR model base on population mobility}.
We now combine the single-city SIR model with the mobile population rules defined in the previous section. 
With the rules of population mobility and the external environment, we can ensure the dynamic balance of population mobility in multiple cities in the model.
By adding the part of population mobility to the single-city model, we get the multi-city model with the following equations:
\begin{align*}
\frac{\mathrm{d}S_{i}}{\mathrm{d}t} & = -\frac{r_{i} \beta I_{i}S_{i}}{N_{i}}+\sum_{j=0,j \neq i}^m f_{i,j}(\frac{S_{j}}{N_{j}}-\frac{S_{i}}{N_{i}}),  \\
\frac{\mathrm{d}I_{i}}{\mathrm{d}t} & = \frac{r_{i} \beta I_{i}S_{i}}{N_{i}}- \gamma I_{i}+\sum_{j=0,j \neq i}^m f_{i,j}(\frac{I_{j}}{N_{j}}-\frac{I_{i}}{N_{i}}),  \\
\frac{\mathrm{d}R_{i}}{\mathrm{d}t} &= \gamma I_{i}+\sum_{j=0,j \neq i}^m f_{i,j}(\frac{R_{j}}{N_{j}}-\frac{R_{i}}{N_{i}}), 
\end{align*}
where $m$ for the total number of cities,
$0$ represents the external environment, and $1 \sim m$ represent the cities;
$f_{i,j}$ represents the floating population between city $i$ and city $j$.
For the change of population quantity under the multi-city model, besides the virus infection within the city given by 
	the term $\frac{r_{i} \beta I_{i}S_{i}}{N_{i}}$ same as in Equation\eqref{eq:single},
	we also account for movements between the cities: for every $j\ne i$, $f_{i,j} S_i/N_i$ is the number of 
	susceptible individuals moving out of
	city $i$ to city $j$, assuming that the mobile population has the same fraction of susceptible individuals as the general city population;
	and $f_{i,j} S_j/N_j$ is the number of susceptible individuals moving into city $i$ from city $j$.
The other equations have the similar explanation.

From the model, in particular the change of $I$, we can see that if city $j$ has a higher portion of infected people
	than city $i$, i.e., $I_j/N_j > I_i / N_i$, then city $j$ contributes the inflow of infected to city $i$, and vice versa.
This matches the intuition of virus propagates from highly infected cities to lowly infected cities.

}

\end{document}